\documentclass[journal]{IEEEtran}
\usepackage{epsfig,latexsym,amsmath,amssymb,setspace,cite,color,graphicx,algorithm,algorithmic,cases}
\usepackage[caption=false,font=footnotesize]{subfig}
\usepackage[percent]{overpic}

\usepackage{amsmath}
\usepackage{amsfonts}
\usepackage{amssymb}
\usepackage{dsfont}
\usepackage{bm}
\usepackage{amsthm}
\usepackage{newlfont}
\usepackage{float}
\usepackage{hyperref}
\usepackage{algorithm}
\usepackage{algorithmic}
\usepackage{enumerate}
\usepackage{chngcntr}
\usepackage{mathtools}
\usepackage{breqn}
\usepackage{stmaryrd}
\usepackage{cite}
\usepackage[yyyymmdd]{datetime}
\hypersetup{
    colorlinks=true, %
    linkcolor= blue, %
    citecolor=blue %
}

\DeclareMathOperator*{\argmax}{arg\,max}

\begin{document}
\sloppy
\allowdisplaybreaks[1]

\newcommand\numberthis{\addtocounter{equation}{1}\tag{\theequation}}

\newtheorem{thm}{Theorem}  
\newtheorem{lem}{Lemma}
\newtheorem{prop}{Proposition}
\newtheorem{cor}{Corollary}
\newtheorem{defn}{Definition}
\newcommand{\remarkend}{\IEEEQEDopen}
\newtheorem{remark}{Remark}
\newtheorem{rem}{Remark}
\newtheorem{ex}{Example}
\newtheorem{pro}{Property}

\newenvironment{example}[1][Example.]{\begin{trivlist}
\item[\hskip \labelsep {\bfseries #1}]}{\end{trivlist}}

\renewcommand{\qedsymbol}{ \begin{tiny}$\blacksquare$ \end{tiny} }

\renewcommand{\leq}{\leqslant}
\renewcommand{\geq}{\geqslant}

\title {Quantifying the Cost of Privately Storing Data in Distributed Storage Systems}

\author{\IEEEauthorblockN{R\'emi A. Chou}

\thanks{
R\'{e}mi A. Chou is with the Department of Electrical Engineering and Computer Science, Wichita State University, Wichita, KS 67260. Part of this work has been presented at the 2022 IEEE International Symposium on Information Theory~(ISIT) \cite{chou2022isit}. E-mail: remi.chou@wichita.edu. This work is supported in part by NSF grant CCF-2047913.}
}
\maketitle
\begin{abstract}
Consider a user who wishes to store a file in multiple servers such that at least $t$ servers are needed to reconstruct the file, and $z$ colluding servers cannot learn any information about the file. Unlike traditional secret-sharing models, where perfectly secure channels are assumed to be available at no cost between the user and each server, we assume that the user can only send data to the servers via a public channel, and that the user and each server share an individual secret key with length $n$. For a given~$n$, we determine the maximal  length of the file that the user can store, and thus quantify the necessary cost to store a file of a certain length, in terms of the length of the secret keys that the user needs to share with the servers. Additionally, for this maximal file length, we determine (i) the optimal amount of local randomness needed at the user, (ii) the optimal amount of public communication from the user to the servers, and (iii) the optimal amount of storage requirement at the servers. %
\end{abstract} 

\section{Introduction}

Centralized data storage of sensitive information could mean compromising the entirety of the
data in the case of a data breach. By contrast, a decentralized storage strategy can offer resilience
against data breaches and avoid having a single point of entry for hackers. Well-known decentralized
strategies are able to ensure that if a file is stored in $L$ servers, then any $t \leq L$ servers that pool their information can reconstruct the file, whereas any $t-1$ compromised  servers do not leak any information about the file in an information-theoretic sense. For instance, secret sharing \cite{shamir1979share,blakley} solves this problem with the optimal storage size requirement at each server. Specifically, to store $F$ bits over $L$ servers, the best possible storage strategy, that allows reconstruction of the information from $t \leq L$ servers and is resilient against data breaches at $t-1$ servers, requires storing ${LF}$ bits over the $L$ servers. In secret sharing models, the user who wishes to store a file in the servers corresponds to the dealer, the file corresponds to a secret, and the information stored at a given server is called a share of the secret. Applications of secret sharing to secure distributed storage have been extensively studied for a wide range of settings, e.g., \cite{rawat2018centralized,agarwal2016security,soleymani2020distributed,huang2016communication,bitar2017staircase,shah2015distributed,huang2016secure,huang2017secret,chou2020secure}.
Note that, as motivated in~\cite{bessani2013depsky,shor2018best,fabian2015collaborative,huang2017secure}, the servers could also correspond to independent cloud storage providers, as it is often less costly for businesses and organizations to outsource data storage but cloud storage providers lack reliable security guarantees and may be the victims of data breaches. %

Since the user and the servers are not physically collocated, a standard assumption in secret sharing models \cite{shamir1979share,blakley,karnin1983secret,mceliece1981sharing,benaloh1988generalized,ito1989secret
} is the availability of individual and information-theoretically secure channels between the user and each server, that allow the user to securely communicate a share of the secret to each server. In this paper, we propose to quantify the cost associated with this assumption. Specifically, instead of assuming the availability at no cost of such information-theoretically secure channels, we assume that the user can communicate over a one-way public channel with each server, and that the user and each server share a secret key, which is a sequence of $n$ bits uniformly distributed over $\{0,1\}^n$. Then, for a given $n$, we determine the maximal length of the file  that the user can store. Given this relationship between $n$ and the maximal length of the file, one can thus determine the necessary cost to store a file of a given  length, in terms of the length of  the secret keys that the user needs to share with the servers. Furthermore, we are also interested in minimizing (i)~the amount of additional resource locally needed at the user, i.e., local randomness needed by the user to form the shares that will be stored at the servers, (ii) the amount of public communication between the user and the servers, and (iii)~the cost of file storage, i.e., the amount of information that needs to be stored at the servers. %

The most challenging part of this study is proving the converse results on the maximal length of the file that the user can store, the optimal amount of local randomness needed at the user, and the optimal amount of public communication between the user and the servers.  Unlike in traditional secret-sharing models, in our converse, we need to account for the presence of shared secret keys, public communication available to all parties, and the fact that \emph{the creation phase of the shares and the secure communication phase of the shares to the servers are allowed to be jointly designed in our model}. Note that  these two phases are \emph{independent} in traditional secret-sharing models, which only focus on the creation phase of the shares since the secure communication phase of the shares relies on the assumption that information-theoretically secure channels are available at no cost. Finally, we establish achievability results that match our converse results using ramp secret sharing schemes~\cite{yamamoto1986secret,blakley1984security}.  
Specifically, we prove the optimality of an achievability scheme that separates the creation of the shares using ramp secret sharing schemes and the secure communication of the shares to the servers via one-time pads.

The remainder of the paper is organized as follows. In Section \ref{secrev}, we describe in an informal manner our setting and the objectives of our study. This section also compares our setting to traditional secret sharing settings and reviews some known results. In Section \ref{sec:ds}, we formally state the problem. In Section \ref{secres}, we present our main results. In Section \ref{sec:ds2}, we extend our setting and results to a multi-user setting. Finally, in Section \ref{seconcl}, we provide concluding remarks.
  \section{Problem Motivation and comparison with traditional secret sharing} \label{secrev}

Consider one user who  wishes to store a file $F$ in $L$ servers, indexed in $\mathcal{L} \triangleq \{ 1, \ldots, L\}$, such that any $t$ servers that pool their information can reconstruct $F$ and any $z$ colluding servers cannot learn any information about $F$, where $t$ and $z$ are chosen in $\{ 1, \ldots, L\}$ and $\{ 1 ,\ldots, t-1\}$, respectively.

\medskip

\textbf{In a traditional secret sharing setting}, the user encodes the file $F$ into $L$ shares $(S_1, \ldots, S_L)$ and transmits the share $S_l$ to Server $l \in \mathcal{L}$ via individual secure channels (available at no cost) between the user and each server. The setting is depicted in Figure \ref{fig1} and the requirements are formalized~as
\begin{align*}
\forall \mathcal{T} \subseteq \mathcal{L}, |\mathcal{T}| \geq t &\implies H(F| S_{\mathcal{T}})=0 \text{ (Recoverability),}\\
\forall \mathcal{U} \subseteq \mathcal{L}, |\mathcal{U}| \leq z &\implies I(F; S_{\mathcal{U}})=0 \text{ (Security),}
\end{align*}
where we have defined $S_{\mathcal{T}} \triangleq (S_l)_{l \in \mathcal{T}}, \forall \mathcal{T} \subseteq \mathcal{L}$. 
\begin{figure}[H]
\centering
\includegraphics[width=0.375\textwidth]{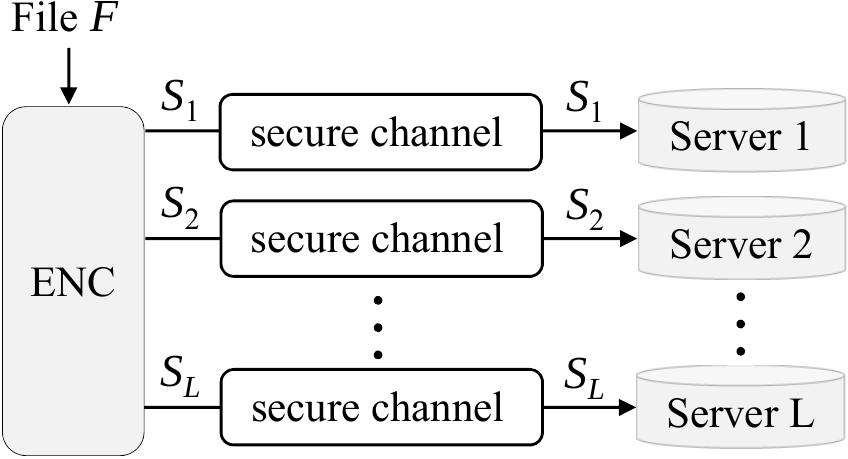}
\caption{Traditional secret sharing settings rely on the assumption that \emph{individual and information-theoretically secure channels between each server and the user  are available at 
no cost}.} \label{fig1}
\end{figure}
In this setting, the following questions arise.
\begin{enumerate}
\item What is the minimum size of an individual share $S_l$, $l \in \mathcal{L}$?
\item What is the minimum size of all the shares $(S_l)_{l \in \mathcal{L}}$ considered jointly?
\item What  is the minimum amount of local randomness needed at the encoder to obtain shares with minimum size?
\end{enumerate}
These questions have all been studied in the literature. It is well known, e.g., \cite{blundo1996randomness,blundo1993efficient}, that designing shares that satisfy $\sum_{l \in \mathcal{L}} H(S_l) =  \frac{L}{t-z}H(F)$ is optimal, and the minimum amount of local randomness needed to achieve this optimal bound is $ \frac{z}{t-z}H(F)$. Moreover, under the additional assumption~that  
$$
\forall \mathcal{T} \subset \mathcal{L},  z < |\mathcal{T}| < t \implies H(F| S_{\mathcal{T}})= \frac{t - |\mathcal{T}|}{t-z}H(F),
$$
then \cite{blundo1993efficient} showed that, for individual shares, having $H(S_l) = \frac{1}{t-z}H(F)$, $l \in \mathcal{L}$,  is optimal.

\smallskip
\textbf{In this paper, we wish to quantify the cost associated with the assumption that individual secure channels are available between the user and each server}. To this end, we replace these individual secure channels by a public channel between the user and the servers and assume that the user shares with Server $l \in \mathcal{L}$ a key $K_l$ with length $n$. The key length $n$ aims to quantify the aforementioned cost. Let $M_l$ be the public communication of the user to Server~$l$, $M \triangleq (M_l)_{l\in \mathcal{L}}$ be the overall public communication, and  $S_l$ be the information stored at Server $l$ after the public communication happened. Our setting is depicted in Figure~\ref{fig2} and the requirements are formalized as 
\begin{align*}
\forall \mathcal{T} \subseteq \mathcal{L}, |\mathcal{T}| \geq t &\implies H(F| S_{\mathcal{T}})=0 \text{ (Recoverability),}\\
\forall \mathcal{U} \subseteq \mathcal{L}, |\mathcal{U}| \leq z &\implies I(F; M, K_{\mathcal{U}})=0 \text{ (Security),}
\end{align*}
where we have defined $S_{\mathcal{T}} \triangleq (S_l)_{l \in \mathcal{T}},$ $K_{\mathcal{T}} \triangleq (K_l)_{l \in \mathcal{T}}, \forall \mathcal{T} \subseteq \mathcal{L}$. Note that the servers may only store a function of the public communication. Additionally, \emph{the creation phase of the shares and the secure communication phase of the shares to the servers are allowed to be jointly designed}, unlike in traditional secret sharing, where the creation of the shares is independent of their secure communication to the servers due to the availability of  secure channels. Note also that  the public communication $M$ now needs to be accounted for information leakage about the file $F$ in the security constraint.
\begin{figure}[H]
\centering
\includegraphics[width=0.44\textwidth]{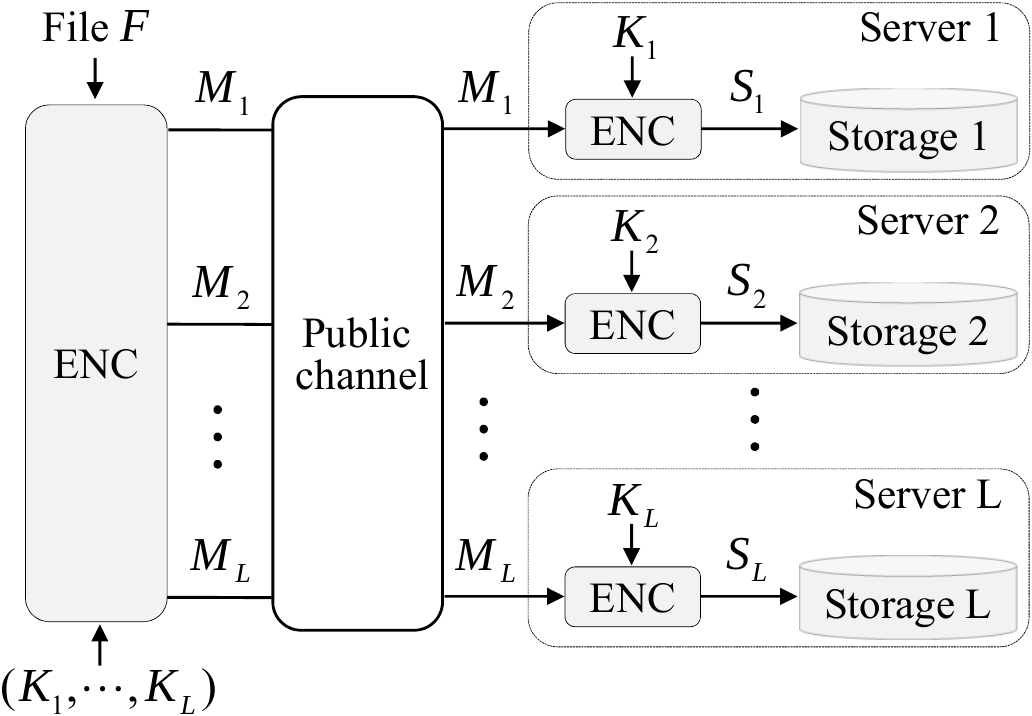}
\caption{In our setting, individual and information-theoretically secure channels between each server and the user are replaced by a public communication channel and pairs of secret keys between the user and the servers. One of our main objectives is to characterize the minimum key lengths needed for a given file size.} \label{fig2}
\end{figure}
In our study we ask the following questions. 
\begin{enumerate}
\item What is the maximal length for the file $F$ that the user can store for a given key length $n$? Let us denote this maximal length by $r^{ (F) }_{\star}$.
\item What is the minimal amount of local randomness needed at the user to achieve $r^{ (F) }_{\star}$?
\item For a given $l \in \mathcal{L}$, what is the minimal storage size needed at Server $l$ for the user to  achieve $r^{ (F) }_{\star}$? In other words, what is the minimal size for $S_l$?
\item For a given $l \in \mathcal{L}$, what is the minimal amount of public communication $M_l$ to Server $l$ needed for the user to achieve $r^{ (F) }_{\star}$?
\item What is the minimal amount of overall public communication $(M_l)_{l \in \mathcal{L}}$ needed for the user to achieve $r^{ (F) }_{\star}$?
\end{enumerate}

\section{Problem statement}\label{sec:ds}
Notation: For any $a,b \in \mathbb{N}^*$, define $\llbracket a, b \rrbracket \triangleq [a,b] \cap \mathbb{N}$. For any $x\in \mathbb{R}$, define $[x]^+\triangleq\max(0,x)$. For a given set $\mathcal{S}$,   let $2^{\mathcal{S}}$  denote the power set of~$\mathcal{S}$. Finally, let $\bigtimes$ denote the Cartesian~product.
\medskip

Consider $L$ servers indexed in $\mathcal{L} \triangleq \llbracket 1, L \rrbracket $ and one user. Assume that Server~$l\in\mathcal{L}$ and the user share a secret key $K_{l}\in \mathcal{K} \triangleq \{0,1\}^{n}$, which is a sequence of $n$ bits uniformly distributed over $\{0,1\}^{n}$. The $L$ keys are assumed to be jointly independent. For any $\mathcal Y \subseteq \mathcal{L}$, we use the notation $K_{\mathcal Y} \triangleq (K_{y})_{ y\in \mathcal Y}$.  

\begin{defn} \label{definition_modelg}
 A~$\left( 2^{r^{(F)}}, 2^{r^{(R)}},\left(2^{r^{(M)}_{l}}\right)_{l \in \mathcal{L}},\left(2^{r^{(S)}_l}\right)_{l \in \mathcal{L}}\right)$ private file storage strategy consists~of 
\begin{itemize}
\item A file $F$ owned by the user, which is uniformly distributed over $\mathcal{F} \triangleq \{ 0,1\}^{r^{(F)}}$ and independent from the keys $K_{\mathcal L}$ (the superscript $(F)$ stands for File);
\item A sequence of local randomness $R$   owned by the user, which is uniformly distributed over $\mathcal{R}  \triangleq \{ 0,1\}^{r^{(R)}}$ and independent from all the other random~variables (the superscript $(R)$ stands for Randomness);
\item $L$ encoding functions 
$h_{l}: \mathcal{R} \times \mathcal{K} \times \mathcal{F}  \to \mathcal{M}_{l},$
where  $l \in \mathcal{L}$, and $\mathcal{M}_{l} \triangleq \{ 0,1\}^{r_{l}^{(M)}}$ (the superscript  $(M)$ stands for  Message);
\item $L$ servers with storage space $r_{l}^{(S)}$ bits for Server~$l\in\mathcal{L}$ (the superscript $(S)$ stands for Server);
\item $L$ encoding functions $g_{l} : \mathcal{M}_{l} \times \mathcal{K} \to \mathcal{S}_{l} ,$ where $l \in \mathcal{L}$,  and $ \mathcal{S}_{l} \triangleq \{ 0,1 \}^{r_{l}^{(S)}}$;
\item $2^L$ decoding functions $\smash{f_{\mathcal{A}} : \displaystyle\bigtimes_{l\in\mathcal{A}} \mathcal{S}_{l} \to \mathcal{F}},$ where $\mathcal{A} \subseteq \mathcal{L}$;
\end{itemize}
and operates as follows:
\begin{enumerate}
\item The user publicly sends to Server $l \in \mathcal{L}$ the message $M_{l} \triangleq h_{l}( R,K_{l},F).$ For $\mathcal Y \subseteq \mathcal{L}$, we define $M_{\mathcal{Y}} \triangleq (M_{l})_{l\in\mathcal{Y}}$. For convenience, we also write $M \triangleq M_{\mathcal{L}}$. 
\item Server $l\in\mathcal{L}$ stores  $S_{l} \triangleq g_{l}( M_{l}, K_{l}).$
\item Any subset of servers $\mathcal{A} \subseteq \mathcal{L}$   can compute $ \widehat{F}(\mathcal{A})\triangleq f_{\mathcal{A}}(S_{\mathcal{A}}),$  an estimate of $F$, where $S_{\mathcal{A}}\triangleq (S_{l})_{l\in \mathcal{A}}$. 
\end{enumerate}
\end{defn}
 The setting is depicted in Figure \ref{fig2}.
\begin{defn} \label{def}
Fix $t  \in \llbracket 1 , L \rrbracket $, $z \in   \llbracket 1 , t -1 \rrbracket$.  Then,  ${r^{(F)}}$ is $(t,z)$-achievable if there exists a $\left(  2^{r^{(F)}} , 2^{r^{(R)}},\left(2^{r^{(M)}_{l}}\right)_{l \in \mathcal{L}},\left(2^{r^{(S)}_l}\right)_{l \in \mathcal{L}}\right)$ private file storage strategy such~that
\begin{align}
	 \forall \mathcal{A} \subseteq \mathcal{L}, |\mathcal{A}| \geq t &\!\implies \! H({F} |\widehat{F} (\mathcal{A})) = 0  \text{ (Recoverability),} \label{eqrel} \\
	\forall \mathcal{U} \subseteq \mathcal{L}, |\mathcal{U}| \leq z &\!\implies  \! I({F}; M, K_{\mathcal{U}} ) = 0  \text{ (Security)}\label{eqSeca}.
\end{align}

The set of all achievable lengths ${r^{(F)}}$ is denoted by $\mathcal{C}_{F}(t,z)$.
\end{defn}
\eqref{eqrel} means that any subset of servers with size larger than or equal to $t$ is able to perfectly recover the files $F$, and  \eqref{eqSeca} means that any subset of servers with size smaller than or equal to $z$ is unable to learn any information about the file. Note that \eqref{eqSeca} accounts for the fact that colluding servers have access to the entire public communication $M$.

Our main objective is to determine, under the constraints~\eqref{eqrel} and \eqref{eqSeca}, the maximal file length that the user can store in the servers given that the secret keys shared with the servers have length~$n$. Next, another of our objectives is to determine   (i)~the minimum amount of local  randomness needed at the user, (ii)~the minimum storage requirement  at the servers, and (iii)~the minimum amount of public communication from the user to the servers that are needed to achieve the largest file rate in $\mathcal{C}_{F}(t,z)$. To this end, we introduce the following~definition.

\begin{defn} \label{def2}
Fix $t  \in \llbracket 1 , L \rrbracket $, $z \in   \llbracket 1 , t\!\!-\!1 \rrbracket$.   For ${r}^{(F)}$ in $\mathcal{C}_{F}(t,z)$, let $\mathcal{Q} ({r}^{(F)})$ be the set of tuples $ T \triangleq  \left(r^{(R)},(r_{l}^{(M)})_{l\in \mathcal{L}},(r_l^{(S)})_{l\in \mathcal{L}} \right)$ such that there exists  a $\left( 2^{r^{(F)}},2^{r^{(R)}},\left(2^{r^{(M)}_{l}}\right)_{l \in \mathcal{L}},\left(2^{r^{(S)}_l}\right)_{l \in \mathcal{L}}\right)$ private file storage strategy that $(t,z)$-achieves ${r}^{(F)}$. Then,   define
\begin{align*}
r^{(F)}_{\star}(t,z) & \triangleq \sup_{{r}^{(F)} \in \mathcal{C}_{F}(t,z)} r^{(F)},\\
 r^{(M)}_{l,\star} (t,z)  &\triangleq \inf_{  T \in \mathcal{Q} ({r}^{(F)}_{\star}(t,z)) }  r_{l}^{(M)},  l \in \mathcal{L},\\
  r^{(M)}_{\Sigma,\star} (t,z)  &\triangleq \inf_{  T \in \mathcal{Q} ({r}^{(F)}_{\star}(t,z) ) } \sum_{l \in \mathcal{L}} r_{l}^{(M)},  \\
  r^{(R)}_{\star} (t,z)  &\triangleq \inf_{  T \in \mathcal{Q} ({r}^{(F)}_{\star}(t,z))}  r^{(R)}, \\
  r^{(S)}_{ l,\star} (t,z)  &\triangleq \inf_{  T \in \mathcal{Q} ({r}^{(F)}_{\star}(t,z)) } r_l^{(S)},   l \in \mathcal{L}.
\end{align*}
\end{defn}
$r^{(F)}_{\star}(t,z)$ is the largest file size that the user can privately store under the constraints~\eqref{eqrel} and \eqref{eqSeca}. Then, $r^{(R)}_{\star}(t,z) $, $r^{(M)}_{l,\star} (t,z)$, $r^{(M)}_{\Sigma,\star} (t,z)$, and $r^{(S)}_{l,\star}(t,z)$, $l \in \mathcal{L}$, are the minimum amount of local randomness, the minimum amount of public communication to Server $l$, the minimum amount of public communication to all the servers, and the minimum storage size required at Server~$l$,  respectively, needed for the user to achieve $r^{(F)}_{\star}(t,z)$. Remark that, a priori, it is unclear whether all these quantities can be achieved simultaneously, i.e., whether there exists  a $\left(  2^{r^{(F)}_{\star}(t,z)} ,2^{r^{(R)}_{\star}(t,z)},\left(2^{r^{(M)}_{l,\star}(t,z)}\right)_{ l \in \mathcal{L}},\left(2^{ r^{(S)}_{l,\star}(t,z)}\right)_{l \in \mathcal{L}}\right)$ file storage strategy that $(t,z)$-achieves $ r^{(F)}_{  \star}(t,z) $.

\section{Main results} \label{secres}
In Section \ref{secimp}, we prove impossibility results. Specifically, we first derive an upper bound on the maximum file length  $r^{(F)}_{\star}(t,z)$. Then, assuming that the user stores a file of length $r^{(F)}_{\star}(t,z)$, we derive lower bounds on the minimum amount of local randomness  $ r^{(R)}_{\star}(t,z)$ needed at the user, the minimum amount of public communication needed to each individual server from the user, i.e., $r^{(M)}_{l,\star}(t,z)$, $l\in\mathcal{L}$, the minimum amount of public communication needed to all the servers from the user, i.e., $r^{(M)}_{\Sigma,\star}(t,z)$, and the minimum storage size needed at Server $l\in\mathcal{L}$, i.e., $ r^{(S)}_{l,\star}(t,z)$. In Section~\ref{secachievres}, we prove an achievability result that matches all the bounds found in Section~\ref{secimp}.

 \subsection{Impossibility results} \label{secimp}
\begin{thm}[Converse on the file length] \label{th1}
Let $t \in \llbracket 1 , L \rrbracket$ and $z \in  \llbracket 1 , t-1 \rrbracket$. Then, we have  
\begin{align*}
r^{(F)}_{\star}(t,z)\leq n (t-z).
\end{align*}
\end{thm} 

\begin{proof}
 Set $D=1$ in Appendix \ref{App_th1}.
\end{proof}

Theorem \ref{th1} means that it is impossible for the user to store a file of length larger than $n (t-z)$ bits. The proof of Theorem \ref{th1} is obtained by first upper bounding the file length by $I( K_{ \mathcal{A}} ;K_{ \mathcal{L}} | K_{ \mathcal{U}} )$  for any $\mathcal{A} ,\mathcal{U} \subseteq \mathcal{L}$ such that $|\mathcal{A}| = t$, $|\mathcal{U}| = z$, and $\mathcal{U} \subset \mathcal{A}$ using Definition \ref{definition_modelg} and the constraints~\eqref{eqrel} and \eqref{eqSeca}. Theorem \ref{th1} is then obtained from this upper bound by leveraging the independence of the secret keys.

\begin{thm}[Converse on storage size requirement at the servers] \label{th2}
Let $t \in \llbracket 1 , L \rrbracket$ and $z \in  \llbracket 1 , t-1 \rrbracket$. Then, we have    
\begin{align*}
 r^{(S)}_{l,\star}(t,z) \geq  n,  \forall l\in \mathcal{L}. 
\end{align*}
\end{thm} 
\begin{proof}
 Set $D=1$ in Appendix \ref{App_th2}.
\end{proof}
Theorem \ref{th2} means that Server $l\in \mathcal{L}$ needs a storage capacity of at least $n$~bits, and is obtained by considering the fact that, at the beginning of the protocol, each server needs to store its secret~key.

\begin{thm}[Converse on the total amount of public communication to the servers]  \label{th3}
Let $t \in \llbracket 1 , L \rrbracket$ and $z \in  \llbracket 1 , t-1 \rrbracket$. Then,  we have  
\begin{align*}
r^{(M)}_{\Sigma,\star}(t,z) \geq \frac{  L }{t-z} r^{(F)}_{\star}(t,z). 
\end{align*}
\end{thm} 
\begin{proof}
 Set $D=1$ in Appendix \ref{App_th3}.
\end{proof}
Theorem \ref{th3} means that it is impossible for the user to store a file of length $r^{(F)}_{\star}(t,z)$ if the public communication sum-length to the servers is smaller than  $\frac{  L }{t-z} r^{(F)}_{\star}(t,z)$ bits. The proof of Theorem \ref{th3} is obtained by first showing that for $\mathcal{T} \subseteq \mathcal{L}$ and $\mathcal{S} \subseteq \mathcal{L} \backslash \mathcal{T} $ such that $|\mathcal{T}|=z$ and $|\mathcal{S}|=t-z$, the sum of the message sizes for the servers in $\mathcal{S}$ is lower bounded by $H(F)$. Then, Theorem \ref{th3} is obtained by a combinatorial argument that consists in summing  this bound over all possible sets of servers $\mathcal{S}$ and $\mathcal{T}$ as above.

\begin{thm}[Converse on the amount of public communication to an individual server] \label{th4}
Let $t \in \llbracket 1 , L \rrbracket$ and $z \in  \llbracket 1 , t-1 \rrbracket$. Consider the following condition
\begin{align}
	 	\forall \mathcal{U},\mathcal{V} \subseteq \mathcal{L}, |\mathcal{U}|=|\mathcal{V}| \implies  I\left({F}; M_{\mathcal{U}},K_{\mathcal{U}}  \right) = I\left({F}; M_{\mathcal{V}},K_{\mathcal{V}}   \right)   \label{eqsym} .
\end{align}
\eqref{eqsym} indicates that any two sets of colluding servers that have the same size have the same amount of information about the file  ${F}$. 
If \eqref{eqsym} holds, then we have  
\begin{align*}
  r^{(M)}_{l,\star} (t,z) \geq \frac{ 1 }{t-z} r^{(F)}_{\star}(t,z), \forall l \in \mathcal{L}. 
\end{align*}
\end{thm} 
\begin{proof}
 Set $D=1$ in Appendix \ref{App_th4}.
\end{proof}
Note that \eqref{eqsym} is always true for sets with cardinality smaller than or equal to $z$ by~\eqref{eqSeca}, and for sets with cardinality larger than or equal to $t$ by \eqref{eqrel}. Note that the concept of leakage symmetry also exists in the context of secret sharing under the denomination uniform secret sharing \cite{yoshida2018optimal}.

Under the leakage symmetry condition \eqref{eqsym}, Theorem \ref{th4} means that it is impossible for the user to store a file of length $r^{(F)}_{ \star}(t,z)$ if the public communication length to Server $l\in\mathcal{L}$ is smaller than  $\frac{  1}{t-z} r^{(F)}_{\star}(t,z)$ bits. Theorem \ref{th4} is obtained by first proving a lower bound on individual public message size that depends on the leakages associated with some sets of servers, specifically, we prove that the public message size for Server $l \in \mathcal{L}$ is lower bounded by $\sum_{i=z  }^{t  -1} [2 \alpha_{i+1} - \alpha_i -\alpha_{i+2} ]^+  $, where for $i \in \mathcal{L}$ and $\mathcal{S} \subseteq \mathcal{L}$ such that $|\mathcal{S}|=i$, we have defined $\alpha_i \triangleq I\left({F} ; M_{\mathcal{S}},K_{\mathcal{S}}  \right)$ and $\alpha_{L+1} \triangleq \alpha_L$. Then, we perform an optimization over all possible values of $(\alpha_i)_{\llbracket 1, L+1 \rrbracket}$ to minimize this bound and obtain Theorem \ref{th4}.

\begin{thm}[Converse on the amount of required local randomness at the users] \label{th5}
Let $t \in \llbracket 1 , L \rrbracket$ and $z \in  \llbracket 1 , t-1 \rrbracket$. Then, we have  
\begin{align*}
   r^{(R)}_{\star} (t,z) \geq \frac{ z }{t-z} r^{(F)}_{\star}(t,z). 
\end{align*}
\end{thm} 

\begin{proof}
 Set $D=1$ in Appendix \ref{App_th5}.
\end{proof}
 Theorem \ref{th5} means that it is impossible for the user to store a file of length $r^{(F)}_{\star}(t,z)$ if the amount of its local randomness is smaller than  $\frac{ z }{t-z} r^{(F)}_{\star}(t,z )$ bits. The proof of Theorem \ref{th5} is obtained by first proving the bound $\sum_{l \in \mathcal{S}} H(M_{l},K_{l} |M_{\mathcal{V}},K_{\mathcal{V}} )  \geq    H( F ) + \sum_{l \in \mathcal{S}}H(K_{l}) $ for $\mathcal{S} \subseteq \mathcal{L} \backslash (\mathcal{T} \cup \mathcal{V}) $ such that $|\mathcal{S}|=t-z$ with $\mathcal{V} \subseteq \mathcal{L}$ such that $|\mathcal{V}|< z$ and $\mathcal{T} \subseteq \mathcal{L} \backslash \mathcal{V}$ such that $|\mathcal{T}|= z - |\mathcal{V}|$. Then, by a combinatorial argument that consists in summing this bound over all possible sets $\mathcal{S}$ and $\mathcal{T}$ as above, we obtain  Theorem \ref{th5}.

 \subsection{Capacity results} \label{secachievres}

We first derive an achievability result with a private file storage strategy that separates the creation of the shares, which is done via ramp secret sharing \cite{yamamoto1986secret,blakley1984security}, and the secure distribution of the shares, which is done via a one-time pad. We will then compare the bounds achieved by this coding strategy with the impossibility results of Section \ref{secimp}, to prove their optimality. 

\begin{thm}\label{th7}
Let $t \in \llbracket 1 , L \rrbracket$ and $z \in  \llbracket 1 , t-1 \rrbracket$. There exists a $\left( 2^{r^{(F)}}, 2^{r^{(R)}},\left(2^{r^{(M)}_{l}}\right)_{l \in \mathcal{L}},\left(2^{r^{(S)}_l}\right)_{l \in \mathcal{L}}\right)$ private file storage strategy that $(t,z)$-achieves $r^{(F)}$ such that 
\begin{align*}
r^{(F)} &= n (t-z) ,\\
r^{(R)} &= nz,\\
r^{(S)}_{l} &= n , \forall l\in \mathcal{L} ,\\
r^{(M)}_{l} & = n  ,  \forall l\in \mathcal{L}. %
\end{align*}
\end{thm}
\begin{proof}
 Set $D=1$ in Appendix \ref{App_th7}.
\end{proof}

From Theorem \ref{th7} and the impossibility results of Section~\ref{secimp}, we obtain a characterization of the quantities introduced in Definition \ref{def2} as follows.  
\begin{thm} \label{th6}
Let $t \in \llbracket 1 , L \rrbracket$ and $z \in  \llbracket 1 , t-1 \rrbracket$. We have
\begin{align*}
r^{(F)}_{ \star} (t,z)&= n (t-z), \\
r^{(R)}_{\star} (t,z) &=  n z,\\
 r^{(S)}_{l,\star}(t,z) &= \textstyle n,  \forall l\in \mathcal{L},\\ 
r^{(M)}_{\Sigma,\star}(t,z) &  = L n, \\
\eqref{eqsym}  \implies  \Big( r^{(M)}_{l,\star} (t,z)&  = n,   \forall l \in \mathcal{L} \Big).  
\end{align*}
\end{thm}
\begin{proof}
 Set $D=1$ in Appendix \ref{App_th6}.
\end{proof}

Note that from Theorem \ref{th7} and Theorem~\ref{th6}, we immediately have Corollary \ref{cor2}, which states that the optimal quantities of Definition~\ref{def2}  can be obtained simultaneously by a single private file storage strategy. 

\begin{cor} \label{cor2}
Let $t \in \llbracket 1 , L \rrbracket$ and $z \in  \llbracket 1 , t-1 \rrbracket$. There exists a $\left( 2^{r^{(F)}}, 2^{r^{(R)}},\left(2^{r^{(M)}_{l}}\right)_{l \in \mathcal{L}},\left(2^{r^{(S)}_l}\right)_{l \in \mathcal{L}}\right)$ private file storage strategy that $(t,z)$-achieves $r^{(F)}$ such that 
\begin{align*}
r^{(F)} &= r^{(F)}_{ \star} (t,z) ,\\
r^{(R)} &= r^{(R)}_{\star} (t,z),\\
r^{(S)}_{l} &= r^{(S)}_{l,\star}(t,z)  , \forall l\in \mathcal{L} ,\\
\textstyle\sum_{l \in\mathcal{L}} r^{(M)}_{l} & =r^{(M)}_{\Sigma,\star}(t,z)  ,\\
r^{(M)}_{l} & =r^{(M)}_{l,\star}(t,z)  ,  \forall l\in \mathcal{L},\text{ when \eqref{eqsym} holds.} 
\end{align*}
\end{cor}

Results interpretation: Consider, $t \in \mathcal{L}$, and $z\in\llbracket 1 ,t-1\rrbracket$. Assume that the user shares an individual key with length $n$ bits with each server. Then, the user can store a file of size at most $n(t-z)$ bits such that any set of servers larger than or equal to $t$ can reconstruct the file, and any set of servers smaller than or equal to $z$ cannot learn anything about the file. Moreover, if the user stores a file of length $n(t-z)$ bits, then the optimal storage capacity at each server is $n$ bits, the optimal amount of local randomness needed at the user is $n\times z$ bits, and the optimal amount of public communication from the user to all the servers is $L\times n$ bits. If one assumes that the leakage about the file must be symmetric among the servers, i.e., \eqref{eqsym} holds, then the optimal  amount of public communication from the user to Server $l\in\mathcal L$ is $ n$ bits.

Note that Corollary \ref{cor2} shows that there is a linear relationship between the maximal length of the file that can be stored and the three resources key length, local randomness, and public communication. This relationship is represented in Figure \ref{figg}.
 \begin{figure}[H]
\centering
\includegraphics[width=0.45\textwidth]{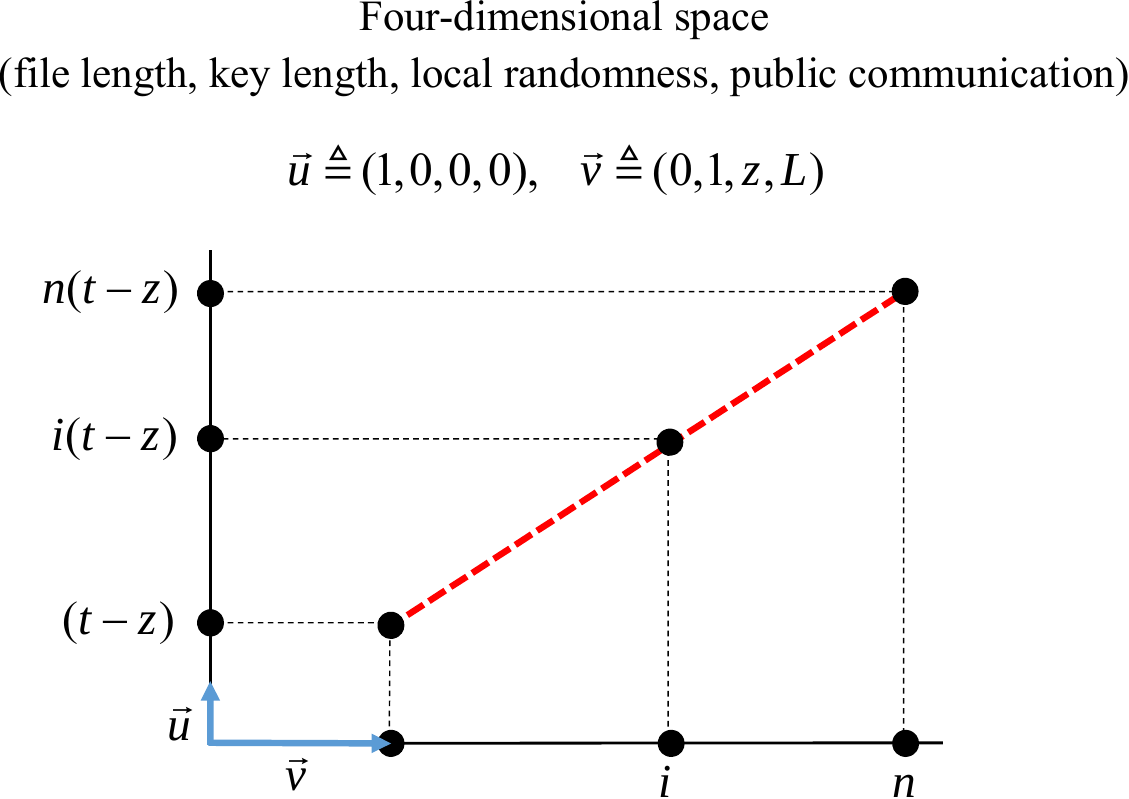}
\caption{Linear relationship between the maximal length of the file that can be stored and the three resources key length, local randomness, and public communication. For instance, for $i \in \mathbb{N}$, storing a file of $i(t-z)$ bits  requires $i$-bit keys, $i \times z$ bits of local randomness, and $i \times L$ bits of public communication.} \label{figg}
\end{figure}

\section{Extension to multiple users}\label{sec:ds2}
In this section, we generalize the problem statement and results of Sections \ref{sec:ds} and \ref{secres}, respectively, to the case where $D\geq 1$ users wish to store files in the servers. %

\subsection{Problem statement}

Consider $L$ servers indexed by $\mathcal{L} \triangleq \llbracket 1, L \rrbracket $ and $D$ users indexed by $\mathcal{D} \triangleq \llbracket 1, D \rrbracket $. Assume that Server~$l\in\mathcal{L}$ and User $d\in\mathcal{D}$ share a secret key $K_{d,l}\in \mathcal{K}_d \triangleq \{0,1\}^{n_d}$, which is a sequence of $n_d$ bits uniformly distributed over $\{0,1\}^{n_d}$. All the $D\times L$ keys are assumed to be jointly independent. For any $\mathcal X \subseteq \mathcal{D}$, $\mathcal Y \subseteq \mathcal{L}$, we define $K_{\mathcal X,\mathcal Y} \triangleq (K_{x,y})_{x\in \mathcal X, y\in \mathcal Y}$ and,  for convenience, we also write $K_{d,\mathcal Y} \triangleq K_{\{d\},\mathcal Y}$ for any $d \in \mathcal{D}$.

\begin{defn} \label{definition_modelg2}
 A $\left( \left(2^{r^{(F)}_d}\right)_{d \in \mathcal{D}}, \left(2^{r^{(R)}_d}\right)_{d \in \mathcal{D}},\left(2^{r^{(M)}_{d,l}}\right)_{d \in \mathcal{D},l \in \mathcal{L}},\right.$ $\left.\left(2^{r^{(S)}_l}\right)_{l \in \mathcal{L}}\right)$ private file storage strategy consists~of 
\begin{itemize}
\item $D$ independent files $F_d$, $d \in \mathcal{D}$, where File $F_d$ is owned by User $d$ and is uniformly distributed over $\mathcal{F}_d \triangleq \{ 0,1\}^{r^{(F)}_d}$. The files $(F_d)_{d \in \mathcal{D}}$ are assumed independent of the keys~$K_{\mathcal D,\mathcal L}$;
\item $D$ independent sequences of local randomness $R_d$, $d \in \mathcal{D}$, where Sequence $R_d$ is owned by User $d$ and is uniformly distributed over $\mathcal{R}_d   \triangleq \{ 0,1\}^{r_d^{(R)}}$. The sequences $(R_d)_{d \in \mathcal{D}}$ are assumed independent of all the other random~variables;
\item $D \times L$ encoding functions 
$h_{d,l}: \mathcal{R}_d \times \mathcal{K}_d \times \mathcal{F}_d  \to \mathcal{M}_{d,l},$
where  $d \in \mathcal{D}$, $l \in \mathcal{L}$, and $\mathcal{M}_{d,l} \triangleq \{ 0,1\}^{r_{d,l}^{(M)}}$;
\item $L$ servers with storage capacities $r_{l}^{(S)}$ bits for Server $l\in\mathcal{L}$;
\item $L \times D$ encoding functions $g_{l,d} : \mathcal{M}_{d,l} \times \mathcal{K}_d \to \mathcal{S}_{l,d} ,$ where $l \in \mathcal{L}$, $d \in \mathcal{D}$, and the $ \mathcal{S}_{l,d}$ are such that $\bigtimes_{d\in\mathcal{D}} \mathcal{S}_{l,d} = \{ 0,1 \}^{r_{l}^{(S)}}$;
\item $2^L \times D$ decoding functions $f_{\mathcal{A},d} : \bigtimes_{l\in\mathcal{A}} \mathcal{S}_{l,d} \to \mathcal{F}_d,$ where $\mathcal{A} \subseteq \mathcal{L}$,  $d \in \mathcal{D}$;
\end{itemize}
and operates as follows:
\begin{enumerate}
\item User $d \in \mathcal{D}$ publicly sends to Server $l \in \mathcal{L}$ the message $M_{d,l} \triangleq h_{d,l}( R_d,K_{d,l},F_d ).$ For $\mathcal X \subseteq \mathcal{D}$, $\mathcal Y \subseteq \mathcal{L}$, we define $M_{\mathcal X,\mathcal{Y}} \triangleq (M_{d,l})_{d\in\mathcal{X},l\in\mathcal{Y}}$ and, for convenience, we also write $M_{\mathcal{X}} \triangleq M_{\mathcal X,\mathcal{L}}$ and $M_{d} \triangleq M_{\{d\},\mathcal{L}}$, $d\in\mathcal{D}$. %
\item Server $l\in\mathcal{L}$ stores $(S_{l,d})_{d \in \mathcal{D}}$ where for $d \in \mathcal{D}$, $S_{l,d} \triangleq g_{l,d}( M_{d,l}, K_{d,l}).$
\item Any subset of servers $\mathcal{A} \subseteq \mathcal{L}$   can compute $ \widehat{F}_d(\mathcal{A})\triangleq f_{\mathcal{A},d}(S_{\mathcal{A},d}),$  an estimate of $F_d$, where $d\in\mathcal{D}$ and $S_{\mathcal{A},d}\triangleq (S_{l,d})_{l\in \mathcal{A}}$.
\end{enumerate}
\end{defn}
\begin{defn} \label{def2b}
Let $\mathbf{t} \triangleq (t_d )_{d\in\mathcal{D}}\in \llbracket 1 , L \rrbracket^D$ and $\mathbf{z} \triangleq (z_d )_{d\in\mathcal{D}} \in \bigtimes_{d \in {\mathcal{D}}} \llbracket 1 , t_d-1 \rrbracket$.  A length-tuple $\left({r^{(F)}_d}\right)_{d \in \mathcal{D}}$ is $(\mathbf{t},\mathbf{z})$-achievable if there exists a $\left( \left(2^{r^{(F)}_d}\right)_{d \in \mathcal{D}}, \left(2^{r^{(R)}_d}\right)_{d \in \mathcal{D}},\left(2^{r^{(M)}_{d,l}}\right)_{d \in \mathcal{D},l \in \mathcal{L}},\left(2^{r^{(S)}_l}\right)_{l \in \mathcal{L}}\right)$ private file storage strategy such~that
\begin{align}
\forall d \in \mathcal{D}, \Big(	 \forall \mathcal{A} \subseteq \mathcal{L}, |\mathcal{A}| \geq t_d &\implies H({F}_{d} |\widehat{F}_{d} (\mathcal{A})) = 0 \Big)  , \label{eqrel2} \\
\forall d \in \mathcal{D}, \Big(	\forall \mathcal{U} \subseteq \mathcal{L}, |\mathcal{U}| \leq z_d &\implies  I({F}_{\mathcal{Z}_d}; M_{\mathcal{D}}, K_{\mathcal{D},\mathcal{U}} ) = 0 \Big)\label{eqSeca2},
\end{align}
where for $d\in \mathcal{D}$, we have defined $\mathcal{Z}_d \triangleq \{ i \in \mathcal{D} : z_i \geq z_d\}$ and $F_{\mathcal{Z}_d} \triangleq (F_i)_{i \in \mathcal{Z}_d}$. The set of all achievable length-tuples is denoted by $\mathcal{C}_{F}(\mathbf{t},\mathbf{z})$.
\end{defn}
\eqref{eqrel2} means that any subset of servers with size larger than or equal to $t_d$ is able to perfectly recover the file $F_{d}$, $d\in\mathcal {D}$, and  \eqref{eqSeca2} means that any subset of servers with size smaller than or equal to $z_d$ is unable to learn any information about the files $\{ F_i : i\in\mathcal{D}, z_i \geq z_d\}$. Similar to the case $D=1$ in Section~\ref{sec:ds}, we have the following counterpart to Definition~\ref{def2}.

\begin{defn} \label{def22}
Let $\mathbf{t} \triangleq (t_d )_{d\in\mathcal{D}}\in \llbracket 1 , L \rrbracket^D$ and $\mathbf{z} \triangleq (z_d )_{d\in\mathcal{D}} \in \bigtimes_{d \in {\mathcal{D}}} \llbracket 1 , t_d-1 \rrbracket$. For $\mathbf{r}^{(F)} \triangleq \left(r^{(F)}_{d}\right)_{d\in \mathcal{D}}$ in $\mathcal{C}_{F}(\mathbf{t},\mathbf{z})$, let $\mathcal{Q} (\mathbf{r}^{(F)})$ be the set of tuples $ T \triangleq  \left((r_d^{(R)})_{d\in \mathcal{D}},(r_{d,l}^{(M)})_{d\in \mathcal{D},l\in \mathcal{L}},(r_l^{(S)})_{l\in \mathcal{L}} \right)$ such that there exists  a $\left( \left(2^{r^{(F)}_d}\right)_{d \in \mathcal{D}}, \left(2^{r^{(R)}_d}\right)_{d \in \mathcal{D}},\left(2^{r^{(M)}_{d,l}}\right)_{d \in \mathcal{D},l \in \mathcal{L}},\left(2^{r^{(S)}_l}\right)_{l \in \mathcal{L}}\right)$ private file storage strategy that $(\mathbf{t},\mathbf{z})$-achieves $\mathbf{r}^{(F)}$. Then, for any $d \in \mathcal{D}$, we  define the following quantities:
\begin{align*}
r^{(F)}_{{d}, \star}(\mathbf{t},\mathbf{z}) & \triangleq \sup_{\mathbf{r}^{(F)} \in \mathcal{C}_{F}(\mathbf{t},\mathbf{z})} r^{(F)}_{{d}}, \\
 r^{(M)}_{d,l,\star} (\mathbf{t},\mathbf{z})  &\triangleq \inf_{ \substack{T \in \mathcal{Q} (\mathbf{r}^{(F)})\\: r^{(F)}_{{d}}=r^{(F)}_{{d}, \star}(\mathbf{t},\mathbf{z}) }}  r_{d,l}^{(M)}, l \in \mathcal{L},\\
  r^{(M)}_{d,\Sigma,\star} (\mathbf{t},\mathbf{z})  &\triangleq \inf_{ \substack{T \in \mathcal{Q} (\mathbf{r}^{(F)})\\: r^{(F)}_{{d}}=r^{(F)}_{{d}, \star}(\mathbf{t},\mathbf{z}) }} \sum_{l \in \mathcal{L}} r_{d,l}^{(M)}, \\
  r^{(R)}_{d,\star} (\mathbf{t},\mathbf{z})  &\triangleq \inf_{ \substack{T \in \mathcal{Q} (\mathbf{r}^{(F)})\\: r^{(F)}_{{d}}=r^{(F)}_{{d}, \star}(\mathbf{t},\mathbf{z}) }}  r_d^{(R)}, \\
  r^{(S)}_{d,l,\star} (\mathbf{t},\mathbf{z})  &\triangleq \inf_{ \substack{T \in \mathcal{Q} (\mathbf{r}^{(F)})\\: r^{(F)}_{{d}}=r^{(F)}_{{d}, \star}(\mathbf{t},\mathbf{z}) }} r_l^{(S)}, l \in \mathcal{L}.
\end{align*}
\end{defn}
$r^{(F)}_{{d}, \star}(\mathbf{t},\mathbf{z})$ is the largest file size that User $d\in \mathcal{D}$ can privately store under the constraints~\eqref{eqrel2} and~\eqref{eqSeca2}. Then, $r^{(R)}_{d,\star}(\mathbf{t},\mathbf{z}) $, $r^{(M)}_{d,l,\star} (\mathbf{t},\mathbf{z})$, $r^{(M)}_{d,\Sigma,\star} (\mathbf{t},\mathbf{z})$, and $r^{(S)}_{d,l,\star}(\mathbf{t},\mathbf{z})$, $l \in \mathcal{L}$, are the minimum amount of local randomness, the minimum amount of public communication to Server $l$, the minimum amount of public communication to all the servers, and the minimum storage size required at Server~$l$,  respectively, needed for User~$d$  to obtain $r^{(F)}_{{d}, \star}(\mathbf{t},\mathbf{z})$. 
A priori, it is unclear whether it is possible to simultaneously obtain $\left(r^{(F)}_{{d}, \star}(\mathbf{t},\mathbf{z})\right)_{d\in\mathcal{D}}$, i.e., whether $\left(r^{(F)}_{{d}, \star}(\mathbf{t},\mathbf{z})\right)_{d\in\mathcal{D}} \in  \mathcal{C}_{F}(\mathbf{t},\mathbf{z})$. We will show that it is actually possible. We will then study whether $ r^{(R)}_{d,\star} (\mathbf{t},\mathbf{z})$, $ r^{(M)}_{d,l,\star}(\mathbf{t},\mathbf{z}) $, and $r^{(S)}_{d,l,\star}(\mathbf{t},\mathbf{z})$, $d \in \mathcal{D}$, $l \in \mathcal{L}$, can  be simultaneously obtained in the  achievability of $\left(r^{(F)}_{{d}, \star}(\mathbf{t},\mathbf{z})\right)_{d\in\mathcal{D}}$ with a single coding scheme.  

 \subsection{Results} \label{secimp2}
In the following theorem, we state impossibility results that are counterpart to the results in Section \ref{secimp} for the case $D=1$.
\begin{thm}
Let $\mathbf{t} \triangleq (t_d )_{d\in\mathcal{D}}\in \llbracket 1 , L \rrbracket^D$ and $\mathbf{z} \triangleq (z_d )_{d\in\mathcal{D}} \in \bigtimes_{d \in {\mathcal{D}}} \llbracket 1 , t_d-1 \rrbracket$. Consider the following leakage symmetry condition
\begin{align}
&	\forall d \in  \mathcal{D}, \left(	 \forall \mathcal{U},\mathcal{V} \subseteq \mathcal{L},  \right. \nonumber \\
	&|\mathcal{U}|=|\mathcal{V}| \implies  \left. I\left({F}_{d}; M_{d,\mathcal{U}},K_{d,\mathcal{U}}  \right) = I\left({F}_{d}; M_{d,\mathcal{V}},K_{d,\mathcal{V}}   \right) \right)  \label{eqsym2} .
\end{align}
Then, for any $d\in \mathcal{D}$, we have  
\begin{align}
r^{(F)}_{{d}, \star}(\mathbf{t},\mathbf{z}) & \leq n_d (t_d-z_d)  ,  \label{th12}\\
r^{(S)}_{d,l,\star}(\mathbf{t},\mathbf{z}) & \geq \sum_{d \in \mathcal D} n_d, \forall l\in \mathcal{L}, \label{th22} \\
r^{(M)}_{d,\Sigma,\star}(\mathbf{t},\mathbf{z}) &\geq \frac{  L }{t_d-z_d} r^{(F)}_{d,\star}(\mathbf{t},\mathbf{z}),  \label{th32} \\
\eqref{eqsym2} \implies  r^{(M)}_{d,l,\star} (\mathbf{t},\mathbf{z}) & \geq \frac{ 1 }{t_d-z_d} r^{(F)}_{d,\star}(\mathbf{t},\mathbf{z}),  \forall l \in \mathcal{L}, \label{th42}\\
   r^{(R)}_{d,\star} (\mathbf{t},\mathbf{z}) &\geq \frac{ z_d }{t_d-z_d} r^{(F)}_{d,\star}(\mathbf{t},\mathbf{z}). \label{th52}
\end{align}
\end{thm}

\eqref{th12} is proved in Appendix \ref{App_th1} and  means that it is impossible for User $d \in \mathcal{D}$ to store a file of length larger than $n_d (t_d-z_d)$ bits. 
 \eqref{th22} is proved in Appendix \ref{App_th2} and means that  Server $l\in \mathcal{L}$ needs a storage capacity of at least $\sum_{d \in \mathcal D} n_d$~bits.
\eqref{th32} is proved in Appendix \ref{App_th3} and means that it is impossible for User $d \in \mathcal{D}$ to store a file of length $r^{(F)}_{{d}, \star}(\mathbf{t},\mathbf{z})$ if the public communication sum-length to the servers is smaller than  $\frac{  L }{t_d-z_d} r^{(F)}_{d,\star}(\mathbf{t},\mathbf{z})$ bits. 
 \eqref{th42} is proved in Appendix \ref{App_th4} and means that, under the leakage symmetry condition \eqref{eqsym2}, it is impossible for User~$d \in \mathcal{D}$ to store a file of length $r^{(F)}_{{d}, \star}(\mathbf{t},\mathbf{z})$ if the public communication length to Server $l\in\mathcal{L}$ is smaller than  $\frac{  1}{t_d-z_d} r^{(F)}_{d,\star}(\mathbf{t},\mathbf{z})$ bits. 
 \eqref{th52} is proved in Appendix \ref{App_th5} and means that it is impossible for User~$d \in \mathcal{D}$ to store a file of length $r^{(F)}_{{d}, \star}(\mathbf{t},\mathbf{z})$ if the amount of its local randomness is smaller than  $\frac{ z_d }{t_d-z_d} r^{(F)}_{d,\star}(\mathbf{t},\mathbf{z})$ bits.

We now give a counterpart to Theorem \ref{th7} for the case $D=1$ with  the following achievability~result.
\begin{thm}\label{th72}
Let $\mathbf{t} \triangleq (t_d )_{d\in\mathcal{D}}\in \llbracket 1 , L \rrbracket^D$ and $\mathbf{z} \triangleq (z_d )_{d\in\mathcal{D}} \in \bigtimes_{d \in {\mathcal{D}}} \llbracket 1 , t_d-1 \rrbracket$. There exists a $\left( \left(2^{r^{(F)}_d}\right)_{d \in \mathcal{D}}, \left(2^{r^{(R)}_d}\right)_{d \in \mathcal{D}},\left(2^{r^{(M)}_{d,l}}\right)_{d \in \mathcal{D},l \in \mathcal{L}},\left(2^{r^{(S)}_l}\right)_{l \in \mathcal{L}}\right)$ private file storage strategy that $(\mathbf{t},\mathbf{z})$-achieves $\left({r^{(F)}_d}\right)_{d \in \mathcal{D}}$ such that, for any $d\in \mathcal{D}$,
\begin{align*}
r^{(F)}_d &= n_d (t_d-z_d),\\
r^{(R)}_d &= n_d z_d, \\
r^{(S)}_{l} &= \textstyle\sum_{d \in \mathcal D} n_d  ,  \forall l\in \mathcal{L} ,\\
r^{(M)}_{d,l} & =n_d ,  \forall l\in \mathcal{L} . 
\end{align*}
\end{thm}
\begin{proof}
See Appendix \ref{App_th7}.
\end{proof}

Then, similar to Theorem \ref{th6} and Corollary \ref{cor2} for the case $D=1$, we deduce the following results. Theorem \ref{th62} provides a characterization of the quantities introduced in Definition~\ref{def22}, and
Corollary~\ref{cor223} shows that the optimal quantities of Definition~\ref{def22} (defined for each individual user) can be obtained simultaneously by a single private file storage strategy.

\begin{thm} \label{th62}
Let $\mathbf{t} \triangleq (t_d )_{d\in\mathcal{D}}\in \llbracket 1 , L \rrbracket^D$ and $\mathbf{z} \triangleq (z_d )_{d\in\mathcal{D}} \in \bigtimes_{d \in {\mathcal{D}}} \llbracket 1 , t_d-1 \rrbracket$. For any $d \in \mathcal{D}$, we have
\begin{align*}
r^{(F)}_{d, \star} (\mathbf{t},\mathbf{z})&= n_d (t_d-z_d),\\
r^{(R)}_{d,\star} (\mathbf{t},\mathbf{z}) &=  n_d z_d, \\
 r^{(S)}_{d,l,\star}(\mathbf{t},\mathbf{z}) &= \textstyle\sum_{d \in \mathcal D} n_d,  \forall l\in \mathcal{L},\\ 
r^{(M)}_{d,\Sigma,\star}(\mathbf{t},\mathbf{z}) &  = L n_d, \\
\eqref{eqsym2}  \implies  \Big( r^{(M)}_{d,l,\star} (\mathbf{t},\mathbf{z})&  = n_d,  \forall l \in \mathcal{L} \Big).  
\end{align*}
\end{thm}
\begin{proof}
See Appendix \ref{App_th6}.
\end{proof}

\begin{cor}\label{cor223}
Let $\mathbf{t} \triangleq (t_d )_{d\in\mathcal{D}}\in \llbracket 1 , L \rrbracket^D$ and $\mathbf{z} \triangleq (z_d )_{d\in\mathcal{D}} \in \bigtimes_{d \in {\mathcal{D}}} \llbracket 1 , t_d-1 \rrbracket$. There exists a $\left( \left(2^{r^{(F)}_d}\right)_{d \in \mathcal{D}}, \left(2^{r^{(R)}_d}\right)_{d \in \mathcal{D}},\left(2^{r^{(M)}_{d,l}}\right)_{d \in \mathcal{D},l \in \mathcal{L}},\left(2^{r^{(S)}_l}\right)_{l \in \mathcal{L}}\right)$ private file storage strategy that $(\mathbf{t},\mathbf{z})$-achieves $\left({r^{(F)}_d}\right)_{d \in \mathcal{D}}$ such that, for any $d \in \mathcal{D}$, 
\begin{align*}
r^{(F)}_d &= r^{(F)}_{d, \star} (\mathbf{t},\mathbf{z}) , \\
r^{(R)}_d &= r^{(R)}_{d,\star} (\mathbf{t},\mathbf{z}), \\
r^{(S)}_{l} &= r^{(S)}_{d,l,\star}(\mathbf{t},\mathbf{z})   ,  \forall l\in \mathcal{L} ,\\
\textstyle\sum_{l \in\mathcal{L}} r^{(M)}_{d,l} & =r^{(M)}_{d,\Sigma,\star}(\mathbf{t},\mathbf{z})  , \\
r^{(M)}_{d,l} & =r^{(M)}_{d,l,\star}(\mathbf{t},\mathbf{z})  ,  \forall l\in \mathcal{L},\text{ when \eqref{eqsym2} holds.} 
\end{align*}
\end{cor}

\section{Concluding remarks} \label{seconcl}
We considered the problem of storing a file in $L$ servers such that any $t \leq L$ servers can reconstruct the file, and any subset of $z<t$ colluding servers cannot learn any information about the file. Unlike solutions that rely on traditional secret sharing models, we developed a new model that does not make the assumption that individual  and information-theoretically secure channels between the user and each server are available at no cost. Instead, we assume that the user can communicate with the servers over a one-way public channel, and share with each server  a secret key with length $n$, which is meant to quantify the cost of privately storing the file. %
For a given secret-key length $n$ and parameters $t$ and $z$, we established the maximal length of the file that  the user can store. Additionally, we determine in this case the minimum amount of local randomness needed at the user, the minimum amount of public communication between the user and the servers, and the minimum amount of storage space required at the servers. While our model allows a joint design of the creation phase of the shares and the secure distribution phase of the shares, our results prove the optimality of an achievability scheme that separates the creation of the shares using ramp secret sharing schemes and the secure distribution of the shares via one-time pads.
Finally, we discussed an extension of our results to a multi-user setting. 

At least two generalizations of the problem setting can be considered and remain open. In the first generalization, the user and the servers communicate over noisy, instead of noiseless, channels; partial results have been obtained in~\cite{zou2015information} for this setting and constructive coding schemes have been proposed in \cite{chou2020unified,chou2018explicit}. In the second generalization, the user and the servers have access to arbitrarily correlated random variables instead of independent and uniformly distributed secret keys; partial results have been obtained in this direction in \cite{rana2021information,chou2021distributed,chou2018secret,csiszar2010capacity} and constructive coding schemes have been proposed~in~\cite{sultana2021}.

\appendices

\section{Proof of Equation \eqref{th12}} \label{App_th1}
Consider an arbitrary  $\left( \left(2^{r^{(F)}_d}\right)_{d \in \mathcal{D}}, \left(2^{r^{(R)}_d}\right)_{d \in \mathcal{D}},\right.$ $\left.\left(2^{r^{(M)}_{d,l}}\right)_{d \in \mathcal{D},l \in \mathcal{L}},\left(2^{r^{(S)}_l}\right)_{l \in \mathcal{L}}\right)$ private file storage strategy that $(\mathbf{t},\mathbf{z})$-achieves $\left({r^{(F)}_d}\right)_{d \in \mathcal{D}}$. Fix $d\in \mathcal{D}$. Define $\{d\}^c \triangleq \mathcal{D} \backslash \{d \}$. In the following lemma, using Definition~\ref{definition_modelg2} and the reliability and security constraints \eqref{eqrel2} and~\eqref{eqSeca2} from Definition~\ref{def2b},  we first give an upper bound on  $r^{(F)}_{d} $ that only depends on the secret keys.
\begin{lem} \label{lem111}
Let $\mathcal{A} ,\mathcal{U} \subseteq \mathcal{L}$ such that $|\mathcal{A}| = t_d$, $|\mathcal{U}| = z_d$, and $\mathcal{U} \subset \mathcal{A}$. We have
\begin{align*}
r^{(F)}_{d} \leq  I( K_{\mathcal{D} ,\mathcal{A}} , K_{\{d\}^c, \mathcal{L}};K_{d, \mathcal{L}} | K_{\mathcal{D}, \mathcal{U}} ).
\end{align*}
\end{lem}  
\begin{proof}
We have
\begin{align*}
&r^{(F)}_{d} \\
& \stackrel{(a)}= H(F_{d})\\
& = H(F_{d} |M_{\mathcal{D}} ,K_{\mathcal{D}, \mathcal{U}} )  + I(F_{d} ;M_{\mathcal{D}}, K_{\mathcal{D}, \mathcal{U}} ) \\
& \stackrel{(b)} \leq H(F_{d} |M_{\mathcal{D}}, K_{\mathcal{D}, \mathcal{U}} )  + I(F_{\mathcal{Z}_d} ;M_{\mathcal{D}} ,K_{\mathcal{D}, \mathcal{U}} )\\
& \stackrel{(c)} =  H(F_{d} |M_{\mathcal{D}}, K_{\mathcal{D}, \mathcal{U}} )\\
& = I(\widehat{F}_{d} (\mathcal{A});F_{d} |M_{\mathcal{D}}, K_{\mathcal{D}, \mathcal{U}} ) + H(F_{d} |M_{\mathcal{D}}, K_{\mathcal{D}, \mathcal{U}},\widehat{F}_{d} (\mathcal{A}))\\
& \stackrel{(d)}\leq I(\widehat{F}_{d} (\mathcal{A});F_{d} |M_{\mathcal{D}}, K_{\mathcal{D}, \mathcal{U}} ) + H(F_{d} |\widehat{F}_{d} (\mathcal{A}))\\
& \stackrel{(e)}= I(\widehat{F}_{d} (\mathcal{A});F_{d} |M_{\mathcal{D}}, K_{\mathcal{D}, \mathcal{U}} )\\
& \stackrel{(f)}\leq  I( M_{\mathcal{D}}, K_{\mathcal{D} ,\mathcal{A}} ;F_{d} |M_{\mathcal{D}}, K_{\mathcal{D}, \mathcal{U}} )\\
& \stackrel{(g)} = I( K_{\mathcal{D} ,\mathcal{A}} ;F_{d} |M_{\mathcal{D}}, K_{\mathcal{D}, \mathcal{U}} )\\
& \stackrel{(h)}\leq  I( K_{\mathcal{D} ,\mathcal{A}}, M_{\{d\}^c};F_{d}, M_{d} | K_{\mathcal{D}, \mathcal{U}} )\\
& \stackrel{(i)}\leq  I( K_{\mathcal{D} ,\mathcal{A}}, F_{\{d\}^c}, K_{\{d\}^c, \mathcal{L}}, R_{\{d\}^c};K_{d, \mathcal{L}}, F_{d}, R_{d} | K_{\mathcal{D}, \mathcal{U}} )\\
& \stackrel{(j)}=  I( K_{\mathcal{D} ,\mathcal{A}}, F_{\{d\}^c}, K_{\{d\}^c, \mathcal{L}} ;K_{d, \mathcal{L}}, F_{d}| K_{\mathcal{D}, \mathcal{U}} ) \\
& \phantom{--}+  I( K_{\mathcal{D} ,\mathcal{A}}, F_{\{d\}^c}, K_{\{d\}^c, \mathcal{L}}; R_{d} | K_{\mathcal{D}, \mathcal{U}} ,K_{d, \mathcal{L}}, F_{d})\\
& \phantom{--}+ I(  R_{\{d\}^c};K_{d, \mathcal{L}}, F_{d}, R_{d} | K_{\mathcal{D}, \mathcal{U}},K_{\mathcal{D} ,\mathcal{A}}, F_{\{d\}^c}, K_{\{d\}^c, \mathcal{L}} ) \\
& \stackrel{(k)}=  I( K_{\mathcal{D} ,\mathcal{A}}, F_{\{d\}^c}, K_{\{d\}^c, \mathcal{L}} ;K_{d, \mathcal{L}}, F_{d} | K_{\mathcal{D}, \mathcal{U}} )\\
& \stackrel{(l)}\leq  I( K_{\mathcal{D} ,\mathcal{A}} , K_{\{d\}^c, \mathcal{L}};K_{d, \mathcal{L}} | K_{\mathcal{D}, \mathcal{U}} ),
\end{align*}
where $(a)$ holds by uniformity of the files $(F_{d})_{d\in \mathcal D}$, $(b)$~holds by the chain rule and non-negativity of the mutual information, $(c)$~holds by \eqref{eqSeca2}  because $|\mathcal{U}|=z_d$, $(d)$ holds because conditioning reduces entropy, $(e)$ holds by \eqref{eqrel2} because $|\mathcal{A}|=t_d$, $(f)$ holds because $\widehat{F}_{d} (\mathcal{A})$ is a function of $S_{\mathcal{A},d}$, which is itself a function of $( M_{\mathcal{D}}, K_{\mathcal{D} ,\mathcal{A}})$, $(g)$~holds because $I( M_{\mathcal{D}};F_{d} |M_{\mathcal{D}}, K_{\mathcal{D}, \mathcal{U}}, K_{\mathcal{D} ,\mathcal{A}}  )=0$, $(h)$~holds by the chain rule applied twice and  non-negativity of the mutual information, $(i)$ holds because $M_{\{d\}^c}$ is a function of $(F_{\{d\}^c}, K_{\{d\}^c, \mathcal{L}}, R_{\{d\}^c})$ (with the notation $F_{\mathcal{S}} \triangleq (F_d)_{d \in \mathcal{S}}$ and $R_{\mathcal{S}} \triangleq (R_d)_{d \in \mathcal{S}}$ for any $\mathcal{S} \subseteq \mathcal{D}$), and $M_{d}$ is a function of $(F_{d}, K_{d, \mathcal{L}}, R_{d})$, $(j)$ holds by the chain rule applied twice, $(k)$ holds by independence between $ R_{\{d\}^c}$ and $( K_{\mathcal{D} ,\mathcal{A}}, F_{\{d\}^c}, K_{\{d\}^c, \mathcal{L}}, K_{d, \mathcal{L}}, F_{d}, R_{d} , K_{\mathcal{D}, \mathcal{U}})$ and by independence between $R_{d}$ and $( K_{\mathcal{D} ,\mathcal{A}}, F_{\{d\}^c}, K_{\{d\}^c, \mathcal{L}}, K_{d, \mathcal{L}}, F_{d}, K_{\mathcal{D}, \mathcal{U}})$, $(l)$ holds by applying twice the chain rule,  independence between $F_{\{d\}^c}$ and $( K_{\mathcal{D} ,\mathcal{A}},  K_{\{d\}^c, \mathcal{L}} ,K_{d, \mathcal{L}}, F_{d} , K_{\mathcal{D}, \mathcal{U}})$, and independence between $F_{d}$ and $( K_{\mathcal{D} ,\mathcal{A}},  K_{\{d\}^c, \mathcal{L}} ,K_{d, \mathcal{L}} , K_{\mathcal{D}, \mathcal{U}})$ similar to $(j)$ and $(k)$.
\end{proof}
Next, we simplify the upper bound of Lemma \ref{lem111} using the independence of the secret keys as follows.
\begin{lem} \label{lem222}
Let $\mathcal{A} ,\mathcal{U} \subseteq \mathcal{L}$ such that $|\mathcal{A}| = t_d$, $|\mathcal{U}| = z_d$, and $\mathcal{U} \subset \mathcal{A}$. We have
\begin{align*}
 I( K_{\mathcal{D} ,\mathcal{A}} , K_{\{d\}^c, \mathcal{L}};K_{d, \mathcal{L}} | K_{\mathcal{D}, \mathcal{U}} ) \leq n_d (t_d-z_d).
\end{align*}
\end{lem}
\begin{proof}
We have
\begin{align*}
& I( K_{\mathcal{D} ,\mathcal{A}} , K_{\{d\}^c, \mathcal{L}};K_{d, \mathcal{L}} | K_{\mathcal{D}, \mathcal{U}} )\\
& = I( K_{\mathcal{D} ,\mathcal{A}} ;K_{d, \mathcal{L}} | K_{\mathcal{D}, \mathcal{U}} ) +  I( K_{\{d\}^c, \mathcal{L}};K_{d, \mathcal{L}} |  K_{\mathcal{D} ,\mathcal{A}}  K_{\mathcal{D}, \mathcal{U}} )\\
& \stackrel{(a)}\leq I( K_{\mathcal{D} ,\mathcal{A}} ;K_{d, \mathcal{L}} | K_{\mathcal{D}, \mathcal{U}} ) \\
& \phantom{--}+  I( K_{\{d\}^c, \mathcal{L}} ,K_{\{d\}^c ,\mathcal{A}},K_{\{d\}^c ,\mathcal{U}};K_{d, \mathcal{L}}, K_{d ,\mathcal{A}} , K_{d ,\mathcal{U}}    )\\
&  = I( K_{\mathcal{D} ,\mathcal{A}} ;K_{d, \mathcal{L}} | K_{\mathcal{D}, \mathcal{U}} ) +  I( K_{\{d\}^c, \mathcal{L}} ;K_{d, \mathcal{L}}    )\\
& \stackrel{(b)} = I( K_{\mathcal{D} ,\mathcal{A}} ;K_{d, \mathcal{L}} | K_{\mathcal{D}, \mathcal{U}} ) \\
& \stackrel{(c)} \leq  I( K_{\mathcal{D} ,\mathcal{A}},K_{\{d\}^c, \mathcal{U}};K_{d, \mathcal{L}} | K_{d, \mathcal{U}} ) \\
& =  I( K_{d ,\mathcal{A}};K_{d, \mathcal{L}} | K_{d, \mathcal{U}} ) \\
& \phantom{--}+ I(  K_{\{d\}^c ,\mathcal{A}} ,K_{\{d\}^c, \mathcal{U}};K_{d, \mathcal{L}} | K_{d, \mathcal{U}} ,K_{d,\mathcal{A}} ) \\
& \leq  I( K_{d ,\mathcal{A}};K_{d, \mathcal{L}} | K_{d, \mathcal{U}} ) \\
& \phantom{--}+ I(  K_{\{d\}^c ,\mathcal{A}} ,K_{\{d\}^c, \mathcal{U}};K_{d, \mathcal{L}} , K_{d, \mathcal{U}} ,K_{d,\mathcal{A}} ) \\
& \stackrel{(d)}= I( K_{d ,\mathcal{A}} ;K_{d, \mathcal{L}} | K_{d, \mathcal{U}} ) \\
& \stackrel{(e)}= H( K_{d ,\mathcal{A}} | K_{d, \mathcal{U}} ) \\
& \stackrel{(f)}= H( K_{d ,\mathcal{A} \backslash \mathcal{U}} ) \\
& \stackrel{(g)}= n_d (t_d-z_d), 
\end{align*}
where  $(a)$ holds by the chain rule applied twice and  non-negativity of the mutual information, $(b)$ holds by independence between the keys $K_{\{d\}^c, \mathcal{L}} $ and $K_{d, \mathcal{L}} $, $(c)$ holds by the chain rule  and  non-negativity of the mutual information, $(d)$ holds by independence between $(K_{\{d\}^c ,\mathcal{A}} ,K_{\{d\}^c, \mathcal{U}})$ and $(K_{d, \mathcal{L}} , K_{d, \mathcal{U}} ,K_{d,\mathcal{A}} )$, $(e)$ holds because $\mathcal{A} \subseteq \mathcal{L}$, $(f)$ holds because $\mathcal{U} \subset \mathcal{A}$, $(g)$ holds because the keys $K_{d,l}$, $l \in \mathcal{A} \backslash \mathcal{U}$, are independent and each uniformly distributed over $\{ 0,1\}^{n_d}$ and $|\mathcal{A} \backslash \mathcal{U}| = t_d-z_d$.
\end{proof}
Next, by combining Lemmas \ref{lem111} and \ref{lem222}, we have
\begin{align*}
r^{(F)}_{d} \leq n_d (t_d-z_d). \numberthis \label{eqconvfilerate}
\end{align*}
Finally, note that \eqref{eqconvfilerate} is valid for any private file storage strategy that $(\mathbf{t},\mathbf{z})$-achieves $\left({r^{(F)}_d}\right)_{d \in \mathcal{D}} \in \mathcal{C}_{F}(\mathbf{t},\mathbf{z})$, 
so, in particular, \eqref{eqconvfilerate} is valid for a  file storage strategy that $(\mathbf{t},\mathbf{z})$-achieves $\left(\tilde{r}^{(F)}_{d}\right)_{d \in \mathcal{D}} $, where $\left({\tilde{r}^{(F)}_{d}}\right)_{d \in \mathcal{D}}  \in \mathcal{C}_{F}(\mathbf{t},\mathbf{z})$ is such that $\tilde{r}^{(F)}_{d} = {r}^{(F)}_{d,\star}(\mathbf{t},\mathbf{z})$.

\section{Proof of Equation \eqref{th22}} \label{App_th2}
Server $l \in\mathcal{L}$ must store the keys $K_{\mathcal{D},l}$ at the beginning of the protocol. Hence, for any $d \in \mathcal D$, $l\in \mathcal{L}$, we must have 
\begin{align*}
r^{(S)}_{d,l,\star}(\mathbf{t},\mathbf{z}) & \geq |K_{\mathcal{D},l}| \\
& =  \sum_{d \in \mathcal D} |K_{d,l}| \\
& =  \sum_{d \in \mathcal D} n_d.
\end{align*}
 
\section{Proof of Equation \eqref{th32}} \label{App_th3}

Consider an arbitrary  $\left( \left(2^{r^{(F)}_d}\right)_{d \in \mathcal{D}}, \left(2^{r^{(R)}_d}\right)_{d \in \mathcal{D}}, \right.$ $\left. \left(2^{r^{(M)}_{d,l}}\right)_{d \in \mathcal{D},l \in \mathcal{L}},\left(2^{r^{(S)}_l}\right)_{l \in \mathcal{L}}\right)$ private file storage strategy that $(\mathbf{t},\mathbf{z})$-achieves $\left({r^{(F)}_d}\right)_{d \in \mathcal{D}}$. Fix $d\in\mathcal{D}$. In the following lemma, using Definition~\ref{definition_modelg2} and the reliability and security constraints \eqref{eqrel2} and~\eqref{eqSeca2} from Definition~\ref{def2b},  we first give a lower bound on the sum of the entropy of the message $(M_{d,l})_{l \in \mathcal{S}}$ for sets $ \mathcal{S} \subset \mathcal{L}$ with cardinality $|\mathcal{S}|=t_d-z_d$.

\begin{lem} \label{lem11}
For $\mathcal{T} \subseteq \mathcal{L}$ and $\mathcal{S} \subseteq \mathcal{L} \backslash \mathcal{T} $ such that $|\mathcal{T}|=z_d$ and $|\mathcal{S}|=t_d-z_d$, we have
 \begin{align*}
   \sum_{l \in \mathcal{S}} H(M_{d,l})   \geq   H( F_{d} )  .
 \end{align*}
\end{lem}

\begin{proof}
We have 
 \begin{align*}
   &\sum_{l \in \mathcal{S}} H(M_{d,l})  + \sum_{l \in \mathcal{S}} H( K_{d,l})\\
  & \stackrel{(a)} \geq \sum_{l \in \mathcal{S}} H(M_{d,l}| K_{d,l})  + \sum_{l \in \mathcal{S}} H( K_{d,l})\\
 & = \sum_{l \in \mathcal{S}} H(M_{d,l},K_{d,l}) \\
 &\stackrel{(b)}  \geq   H(M_{d,\mathcal{S}},K_{d,\mathcal{S}})\\
 & \stackrel{(c)}\geq H(M_{d,\mathcal{S}},K_{d,\mathcal{S}} | M_{d,\mathcal{T}},K_{d,\mathcal{T}})  
 \\
 & =    I(M_{d,\mathcal{S}},K_{d,\mathcal{S}}; F_{d} | M_{d,\mathcal{T}},K_{d,\mathcal{T}}) \\
& \phantom{--}+ H(M_{d,\mathcal{S}},K_{d,\mathcal{S}} | F_{d} ,M_{d,\mathcal{T}},K_{d,\mathcal{T}}) \\
 & =   H( F_{d} | M_{d,\mathcal{T}},K_{d,\mathcal{T}})- H(F_{d} |M_{d,\mathcal{S}},K_{d,\mathcal{S}}, M_{d,\mathcal{T}},K_{d,\mathcal{T}})  \\
& \phantom{--}+ H(M_{d,\mathcal{S}},K_{d,\mathcal{S}} | F_{d} ,M_{d,\mathcal{T}},K_{d,\mathcal{T}}) \\
 & \stackrel{(d)} \geq  H( F_{d} | M_{d,\mathcal{T}},K_{d,\mathcal{T}})- H(F_{d} |S_{\mathcal{S} \cup \mathcal T,d} )  \\
& \phantom{--}+ H(M_{d,\mathcal{S}},K_{d,\mathcal{S}} | F_{d} ,M_{d,\mathcal{T}},K_{d,\mathcal{T}}) \\
  & \stackrel{(e)}\geq  H( F_{d} | M_{d,\mathcal{T}},K_{d,\mathcal{T}})- H(F_{d} |\widehat F_d(\mathcal{S} \cup \mathcal T) )  \\
& \phantom{--}+ H(M_{d,\mathcal{S}},K_{d,\mathcal{S}} | F_{d} ,M_{d,\mathcal{T}},K_{d,\mathcal{T}}) \\
    & = H(F_d) - I( F_{d} ; M_{d,\mathcal{T}},K_{d,\mathcal{T}})- H(F_{d} |\widehat F_d(\mathcal{S} \cup \mathcal T) )  \\
& \phantom{--}+ H(M_{d,\mathcal{S}},K_{d,\mathcal{S}} | F_{d} ,M_{d,\mathcal{T}},K_{d,\mathcal{T}}) \\
        & \stackrel{(f)}\geq H(F_d) - I( F_{\mathcal{Z}_d} ; M_{\mathcal{D}},K_{\mathcal{D},\mathcal{T}})- H(F_{d} |\widehat F_d(\mathcal{S} \cup \mathcal T) )  \\
& \phantom{--}+ H(M_{d,\mathcal{S}},K_{d,\mathcal{S}} | F_{d} ,M_{d,\mathcal{T}},K_{d,\mathcal{T}}) \\
    & \stackrel{(g)} =  H( F_{d} ) + H(M_{d,\mathcal{S}},K_{d,\mathcal{S}} | F_{d} ,M_{d,\mathcal{T}},K_{d,\mathcal{T}}) \\
        & \stackrel{(h)} \geq  H( F_{d} ) + H(M_{d,\mathcal{S}},K_{d,\mathcal{S}} | F_{d} ,R_{d},K_{d,\mathcal{T}}) \\
                & \stackrel{(i)}  \geq  H( F_{d} ) + H(K_{d,\mathcal{S}} | F_{d} ,R_{d},K_{d,\mathcal{T}}) \\
                                &  \stackrel{(j)}=  H( F_{d} ) + \sum_{l \in \mathcal{S}}H(K_{d,l}) , 
 \end{align*}
 where $(a)$ and $(c)$ hold because conditioning reduces entropy, $(b)$ holds by the chain rule and because conditioning reduces entropy, $(d)$ holds because $S_{\mathcal{S},d}$ is a function of $(M_{d,\mathcal{S}},K_{d,\mathcal{S}})$  and  $S_{\mathcal{T},d}$ is a function of $(M_{d,\mathcal{T}},K_{d,\mathcal{T}})$, $(e)$ holds because $\widehat F_d(\mathcal{S} \cup \mathcal T) $ is a function of $S_{\mathcal{S} \cup \mathcal T,d} $, $(f)$ holds by the chain rule and non-negativity of the mutual information, $(g)$ holds by~\eqref{eqSeca2} because $|\mathcal{T}|=z_d$ and by \eqref{eqrel2} because  $|\mathcal{S} \cup \mathcal{T}|=t_d$, $(h)$ holds because $M_{d,\mathcal{T}}$ is a function of $(F_{d} ,R_{d},K_{d,\mathcal{T}})$, $(i)$ holds by the chain rule and non-negativity of the entropy, $(j)$ holds because $K_{d,\mathcal{S}}$ is independent from $( F_{d} ,R_{d},K_{d,\mathcal{T}})$ (since $\mathcal{S} \cap \mathcal{T} = \emptyset$) and because the keys $(K_{d,l})_{l \in \mathcal{S}}$ are independent.
\end{proof}
 Next, by summing both sides of the equation of Lemma \ref{lem11} over all possible sets $\mathcal{T} \subseteq \mathcal{L}$ and $\mathcal{S} \subseteq \mathcal{L} \backslash \mathcal{T} $ such that $|\mathcal{T}|=z_d$ and $|\mathcal{S}|=t_d-z_d$, we obtain a lower bound on the sum of the entropy of all the message $(M_{d,l})_{l \in \mathcal{L}}$.
 
 \begin{lem} \label{lem22}
 We have 
 \begin{align*}
 \sum_{l \in   \mathcal{L}} H(M_{d,l}) \geq \frac{L}{t_d-z_d} H( F_d ).
 \end{align*}
 \end{lem}
\begin{proof}
We have
   \begin{align*} 
& \frac{L}{t_d-z_d} H( F_{d} ) \\
 & \stackrel{(a)}= \frac{L}{t_d-z_d} \Upsilon_d \sum_{ \substack{\mathcal{T} \subseteq \mathcal{L}  \\  |\mathcal{T} | = z_d }} \sum_{ \substack{\mathcal{S} \subseteq  \mathcal{T}^c \\  |\mathcal{S} | = t_d-z_d }}  H( F_{d} ) \\ 
   &  \stackrel{(b)} \leq \frac{L}{t_d-z_d} \Upsilon_d \sum_{ \substack{\mathcal{T} \subseteq \mathcal{L}  \\  |\mathcal{T} | = z_d }} \sum_{ \substack{\mathcal{S} \subseteq \mathcal{T}^c  \\  |\mathcal{S} | = t_d-z_d }}  \sum_{l \in \mathcal{S}} H(M_{d,l})  \\
   & \stackrel{(c)} = \frac{L}{t_d-z_d} \Upsilon_d \sum_{ \substack{\mathcal{T} \subseteq \mathcal{L}  \\  |\mathcal{T} | = z_d }} { L-z_d - 1 \choose t_d-z_d-1 }  \sum_{l \in  \mathcal{T}^c } H(M_{d,l})  \\
   &  \stackrel{(d)}  = \frac{L}{t_d-z_d} \Upsilon_d  { L-z_d - 1 \choose t_d-z_d-1 } \sum_{ \substack{\mathcal{T} \subseteq \mathcal{L}  \\  |\mathcal{T} | =L- z_d }}   \sum_{l \in   \mathcal{T}} H(M_{d,l})  \\
      & \stackrel{(e)}= \frac{L}{t_d-z_d} \Upsilon_d { L-z_d - 1 \choose t_d-z_d-1 } { L - 1 \choose L-z_d-1 }   \sum_{l \in   \mathcal{L}} H(M_{d,l})  \\
            & =  \sum_{l \in   \mathcal{L}} H(M_{d,l})  ,
 \end{align*}
 where  $(a)$ holds with $\Upsilon_d \triangleq {L \choose z_d}^{-1} {{L-z_d}\choose{t_d-z_d}}^{-1}$, $(b)$ holds by Lemma \ref{lem11}, $(c)$ holds because for any $l\in\mathcal{T}^c$, $H(M_{d,l})$ appears exactly ${ L-z_d - 1 \choose t_d-z_d-1 }$ times in the term $\sum_{ \substack{\mathcal{S} \subseteq \mathcal{T}^c  \\  |\mathcal{S} | = t_d-z_d }}  \sum_{l \in \mathcal{S}} H(M_{d,l})$ (note that this observation was also made in \cite[Lemma 3.2]{de1999multiple}), $(d)$ holds by a change of variables in the sums, $(e)$ holds because for any $l\in\mathcal{L}$, $H(M_{d,l})$ appears exactly ${ L - 1 \choose L-z_d-1 }$ times in the term $\sum_{ \substack{\mathcal{T} \subseteq \mathcal{L}  \\  |\mathcal{T} | =L- z_d }}   \sum_{l \in   \mathcal{T}} H(M_{d,l})$. 
 \end{proof}

Finally, we have
   \begin{align*} 
  \frac{L}{t_d-z_d}r^{(F)}_{d}  
 & \stackrel{(a)}  =  \frac{L}{t_d-z_d} H( F_{d} )\\
             & \stackrel{(b)} \leq  \sum_{l \in   \mathcal{L}} H(M_{d,l})  \\
            &  \leq  \sum_{l \in   \mathcal{L}} r^{(M)}_{d,l},  \numberthis \label{eqratecsum}
 \end{align*} 
 where  $(a)$ holds by uniformity of $F_d$, $(b)$~holds by Lemma \ref{lem22}.

 Since \eqref{eqratecsum} is valid for any private file storage strategy, \eqref{eqratecsum} is also valid for a  $\left( \left(2^{\tilde r^{(F)}_{d}}\right)_{d \in \mathcal{D}},\left(2^{\tilde r^{(R)}_{d}}\right)_{d \in \mathcal{D}},\left(2^{\tilde r^{(M)}_{d,l,}}\right)_{d \in \mathcal{D},l \in \mathcal{L}},\left(2^{\tilde r^{(S)}_{l}}\right)_{l \in \mathcal{L}}\right)$ file storage strategy that $(\mathbf{t},\mathbf{z})$-achieves $\left(\tilde{r}^{(F)}_{d}\right)_{d \in \mathcal{D}} $, where $\left({\tilde{r}^{(F)}_{d}}\right)_{d \in \mathcal{D}}  \in \mathcal{C}_{F}(\mathbf{t},\mathbf{z})$ is such that $\tilde{r}^{(F)}_{d} = {r}^{(F)}_{d,\star}(\mathbf{t},\mathbf{z})$ and $\sum_{l\in\mathcal{L}} \tilde{r}^{(M)}_{d,l}=r^{(M)}_{d,\Sigma,\star}(\mathbf{t},\mathbf{z}) $.

\section{Proof of Equation \eqref{th42}} \label{App_th4}

Consider an arbitrary  $\left( \left(2^{r^{(F)}_d}\right)_{d \in \mathcal{D}}, \left(2^{r^{(R)}_d}\right)_{d \in \mathcal{D}},\right.$ $\left.\left(2^{r^{(M)}_{d,l}}\right)_{d \in \mathcal{D},l \in \mathcal{L}},\left(2^{r^{(S)}_l}\right)_{l \in \mathcal{L}}\right)$ private file storage strategy that $(\mathbf{t},\mathbf{z})$-achieves $\left({r^{(F)}_d}\right)_{d \in \mathcal{D}}$. Assume that \eqref{eqsym2} holds. Fix $d\in\mathcal{D}$, $l \in \mathcal{L}$. By exploiting the leakage symmetry condition \eqref{eqsym2}, we derive a first lower bound on the public communication to a specific server in the following lemma.
\begin{lem} \label{lem444}
 For $i \in \llbracket z_d , t_d -1\rrbracket$, define $\mathcal{V}_i \triangleq \begin{cases} \llbracket 1,i\rrbracket & \text{ if } l >i \\ \llbracket 1,i +1 \rrbracket \backslash \{ l \} & \text{ if } l \leq i \end{cases}$  and $\mathcal{V}_{t_d} \triangleq \mathcal{V}_{t_d-1} \cup \{ l\}$.  For $i \in \mathcal{L}$, and $\mathcal{S} \subseteq \mathcal{L}$ such that $|\mathcal{S}|=i$, define $\alpha_i \triangleq I\left({F}_{d}; M_{d,\mathcal{S}},K_{d,\mathcal{S}}  \right)$ and $\alpha_{L+1} \triangleq \alpha_L$. Note that $\alpha_i$ only depends on $i$ and not on the specific elements of $\mathcal{S}$ by \eqref{eqsym2}. Note also that $\alpha_{z_d} =0$ by~\eqref{eqSeca2} and $\alpha_{t_d} = H(F_d)$ by~\eqref{eqrel2}.  Then, we have
\begin{align*}
H(M_{d,l})\geq   \sum_{i=z_d}^{t_d-1} [2 \alpha_{i+1} - \alpha_i -\alpha_{i+2} ]^+. \numberthis \label{eqmin}
	\end{align*}
\end{lem}
\begin{proof}
We have
\begin{align*}
&	H(M_{d,l}) + H(K_{d,l})\\
& \stackrel{(a)}\geq 	H(M_{d,l}) + H(K_{d,l}|M_{d,l})\\
&	= H(M_{d,l},K_{d,l})\\
	& \stackrel{(b)}\geq H(M_{d,l},K_{d,l}|M_{d,\mathcal{V}_{z_d}},K_{d,\mathcal{V}_{z_d}})    \\ \nonumber
& \stackrel{(c)}= H(M_{d,l},K_{d,l}|M_{d,\mathcal{V}_{z_d}},K_{d,\mathcal{V}_{z_d}}) \nonumber \\
& \phantom{--}- H(M_{d,l},K_{d,l}| M_{d,\mathcal{V}_{t_d}},K_{d,\mathcal{V}_{t_d}})\\ \nonumber
	& = \smash\sum_{i=z_d}^{t_d-1} \left[H(M_{d,l},K_{d,l}|M_{d,\mathcal{V}_{i}},K_{d,\mathcal{V}_{i}}) \nonumber \right. \\
& \phantom{----} \left. - H(M_{d,l},K_{d,l}| M_{d,\mathcal{V}_{i+1}},K_{d,\mathcal{V}_{i+1}}) \right] \numberthis \label{eqint1}\\ \nonumber 
		& \stackrel{(d)}= \smash\sum_{i=z_d}^{t_d-1} \left[ H(M_{d,l},K_{d,l},F_d|M_{d,\mathcal{V}_{i}},K_{d,\mathcal{V}_{i}})\right. \\
& \phantom{----} \left. - H(F_d|M_{d,\mathcal{V}_{i}\cup \{l\}},K_{d,\mathcal{V}_{i}\cup \{l\}}) \right.\\
		& \phantom{----} \left. - H( M_{d,l},K_{d,l},F_d| M_{d,\mathcal{V}_{i+1}},K_{d,\mathcal{V}_{i+1}}) \right. \\
& \phantom{----} \left.+ H(F_d| M_{d,\mathcal{V}_{i+1} \cup \{l\}},K_{d,\mathcal{V}_{i+1} \cup \{l\}}) \right] \\ \nonumber 
				& \stackrel{(e)}= \smash\sum_{i=z_d}^{t_d-1}  \left[H(F_d|M_{d,\mathcal{V}_{i}},K_{d,\mathcal{V}_{i}}) \right. \\
& \phantom{----} \left.+ H(M_{d,l},K_{d,l}|F_d,M_{d,\mathcal{V}_{i}},K_{d,\mathcal{V}_{i}})\right. \\
& \phantom{----} \left.- H(F_d|M_{d,\mathcal{V}_{i}\cup \{l\}},K_{d,\mathcal{V}_{i}\cup \{l\}})\right.\\
		& \phantom{----} \left.- H( F_d| M_{d,\mathcal{V}_{i+1}},K_{d,\mathcal{V}_{i+1}})\right. \\
& \phantom{----} \left. - H( M_{d,l},K_{d,l}|F_d, M_{d,\mathcal{V}_{i+1}},K_{d,\mathcal{V}_{i+1}}) \right.\\
		& \phantom{----} \left.+ H(F_d| M_{d,\mathcal{V}_{i+1} \cup \{l\}},K_{d,\mathcal{V}_{i+1} \cup \{l\}}) \right] \\ \nonumber 
			& = \smash\sum_{i=z_d}^{t_d-1} \left[- I(F_d;M_{d,\mathcal{V}_{i}},K_{d,\mathcal{V}_{i}}) \right. \\
& \phantom{----} \left.+ H(M_{d,l},K_{d,l}|F_d,M_{d,\mathcal{V}_{i}},K_{d,\mathcal{V}_{i}})\right. \\
& \phantom{----} \left.+I(F_d;M_{d,\mathcal{V}_{i}\cup \{l\}},K_{d,\mathcal{V}_{i}\cup \{l\}}) \right.\\
		& \phantom{----} \left.+I(F_d; M_{d,\mathcal{V}_{i+1}},K_{d,\mathcal{V}_{i+1}}) \right. \\
& \phantom{----} \left.- H( M_{d,l},K_{d,l}|F_d, M_{d,\mathcal{V}_{i+1}},K_{d,\mathcal{V}_{i+1}})\right.\\
		& \phantom{----} \left. - I(F_d; M_{d,\mathcal{V}_{i+1} \cup \{l\}},K_{d,\mathcal{V}_{i+1} \cup \{l\}}) \right] \\ \nonumber 
			& \stackrel{(f)}= \smash\sum_{i=z_d}^{t_d-1} \left[ 2 \alpha_{i+1} - \alpha_i -\alpha_{i+2}  \right.\\
		& \phantom{----} \left. + H(M_{d,l},K_{d,l}|F_d,M_{d,\mathcal{V}_{i}},K_{d,\mathcal{V}_{i}})  \right. \\
& \phantom{----} \left.- H( M_{d,l},K_{d,l}|F_d, M_{d,\mathcal{V}_{i+1}},K_{d,\mathcal{V}_{i+1}}) \right] \\ \nonumber 
										& \stackrel{(g)} \geq   [ 2 \alpha_{t_d} - \alpha_{t_d-1} -\alpha_{t_d+1}\\
& \phantom{----}+H(M_{d,l},K_{d,l}|F_d,M_{d,\mathcal{V}_{t_d-1}},K_{d,\mathcal{V}_{t_d-1}})] \\
& \phantom{----}+ \smash\sum_{i=z_d}^{t_d-2}[ 2 \alpha_{i+1} - \alpha_i -\alpha_{i+2}]    \\ \nonumber
																		& \stackrel{(h)} \geq   [ 2 \alpha_{t_d} - \alpha_{t_d-1} -\alpha_{t_d+1}\\
& \phantom{----}+H(M_{d,l},K_{d,l}|F_d,M_{d,\mathcal{V}_{t_d-1}},K_{d,\mathcal{V}_{t_d-1}})]^+ \\
& \phantom{----}+ \smash\sum_{i=z_d}^{t_d-2}[ 2 \alpha_{i+1} - \alpha_i -\alpha_{i+2}]^+    \\ \nonumber
																		& \stackrel{(i)} =   [ \alpha_{t_d} - \alpha_{t_d-1} +H(M_{d,l},K_{d,l}|F_d,M_{d,\mathcal{V}_{t_d-1}},K_{d,\mathcal{V}_{t_d-1}})]^+ \\
		& \phantom{----}+ \smash\sum_{i=z_d}^{t_d-2}[ 2 \alpha_{i+1} - \alpha_i -\alpha_{i+2}]^+    \\ \nonumber 																& \stackrel{(j)} =   [ \alpha_{t_d} - \alpha_{t_d-1} +H(M_{d,l},K_{d,l}|F_d,M_{d,\mathcal{V}_{t_d-1}},K_{d,\mathcal{V}_{t_d-1}})] \\
		& \phantom{----}+ \smash\sum_{i=z_d}^{t_d-2}[ 2 \alpha_{i+1} - \alpha_i -\alpha_{i+2}]^+    \\ \nonumber 
										& \stackrel{(k)} =   H(M_{d,l},K_{d,l}|F_d,M_{d,\mathcal{V}_{t_d-1}},K_{d,\mathcal{V}_{t_d-1}})  \\
		& \phantom{----}+ \smash\sum_{i=z_d}^{t_d-1}[ 2 \alpha_{i+1} - \alpha_i -\alpha_{i+2}]^+    \\ \nonumber 
										& \stackrel{(l)} \geq   H(K_{d,l}|F_d,M_{d,\mathcal{V}_{t_d-1}},K_{d,\mathcal{V}_{t_d-1}}) \\
		& \phantom{----}+ \smash\sum_{i=z_d}^{t_d-1} [2 \alpha_{i+1} - \alpha_i -\alpha_{i+2}]^+    \\ \nonumber 
																			& \stackrel{(m)} \geq   H(K_{d,l}|F_d,R_d ,K_{d,\mathcal{V}_{t_d-1}}) + \sum_{i=z_d}^{t_d-1} [2 \alpha_{i+1} - \alpha_i -\alpha_{i+2} ]^+    \\ \nonumber 								
		& \stackrel{(n)} =  H(K_{d,l} ) + \smash\sum_{i=z_d}^{t_d-1} [2 \alpha_{i+1} - \alpha_i -\alpha_{i+2} ]^+   , 
	\end{align*}
	where $(a)$ and $(b)$ hold  because conditioning reduces entropy, $(c)$ holds because $l \in \mathcal{V}_{t_d}$, $(d)$ and $(e)$ hold  by the chain rule, $(f)$ holds by the definition of $\alpha_i$, $(g)$ holds because for any $i\in \llbracket  z_d , t_d-2 \rrbracket$, $H(M_{d,l},K_{d,l}|F_d,M_{d,\mathcal{V}_{i}},K_{d,\mathcal{V}_{i}})  \geq H( M_{d,l},K_{d,l}|F_d, M_{d,\mathcal{V}_{i+1}},K_{d,\mathcal{V}_{i+1}})$ since conditioning reduces entropy and $\mathcal{V}_{i} \subset \mathcal{V}_{i+1}$, and because $H(M_{d,l},K_{d,l}|F_d,M_{d,\mathcal{V}_{t_d}},K_{d,\mathcal{V}_{t_d}}) =0 $ since $l \in \mathcal{V}_{t_d}$, $(h)$~holds because in \eqref{eqint1}, we observe that $H(M_{d,l},K_{d,l}|M_{d,\mathcal{V}_{i}},K_{d,\mathcal{V}_{i}}) - H(M_{d,l},K_{d,l}| M_{d,\mathcal{V}_{i+1}},K_{d,\mathcal{V}_{i+1}})\geq 0$ since conditioning reduces entropy and $\mathcal{V}_{i} \subset \mathcal{V}_{i+1}$, $(i)$ holds because $\alpha_{t_d+1} = \alpha_{t_d} = H(F_d)$ by~\eqref{eqrel2}, $(j)$~holds because $ \alpha_{t_d} \geq \alpha_{t_d-1}$ by the definition of $ \alpha_{t_d}$ and $\alpha_{t_d-1}$, $(k)$ holds because $ \alpha_{t_d} - \alpha_{t_d-1}= [2\alpha_{t_d} - \alpha_{t_d-1}-\alpha_{t_d+1} ]^+$, $(l)$ holds by the chain rule and non-negativity of the entropy, $(m)$~holds because $M_{d,\mathcal{V}_{t_d-1}}$ is a function of $(F_{d} ,R_{d},K_{d,\mathcal{V}_{t_d-1}})$, $(n)$~holds by independence between $K_{d,l}$ and $(F_d,R_d ,K_{d,\mathcal{V}_{t_d-1}})$ since $\{l\}\cap \mathcal{V}_{t_d-1}=\emptyset$.
	\end{proof}

Next, we remark that the lower bound of Lemma \ref{lem444} is lower bounded by
$$
\min_{f \in \mathcal{F}} \sum_{i=1}^{t_d-z_d}  [ 2 f(i+1)- f(i)  - f(i+2)  ]^+ ,
$$
where the minimum is taken over the set $\mathcal{F}$ of  all the functions $f: \llbracket 1,t_d-z_d+2 \rrbracket \to [0,H(F_d)]$ that are non-decreasing (because, by construction, $(\alpha_i)_{i\in \llbracket 1, L+1 \rrbracket}$ is a non-decreasing sequence) and such that $f(1)  = \alpha_{z_d} = 0$, $ f(t_d-z_d+2)=f(t_d-z_d+1)= \alpha_{t_d} =H(F_d)$. In the following lemma, we determine a lower bound for this optimization problem.

\begin{lem} \label{lem555}
For any $f \in \mathcal{F}$, we have
\begin{align}
\sum_{i=1}^{t_d-z_d}  [ 2 f(i+1)- f(i)  - f(i+2)  ]^+ \geq  \frac{H(F_d)}{t_d-z_d} . \label{eqf1}
\end{align}
\end{lem}

\begin{proof}
	 Let $f \in \mathcal{F}$ and let $f^+$ be the concave envelope of $f$ over $\llbracket 1 ,t_d-z_d+2 \rrbracket$, i.e., for $i \in \llbracket 1 ,t_d-z_d+2 \rrbracket$, $f^+(i) \triangleq \min \{ g(i) : g \geq f, g \text{ is concave}\}$. Note that $f^+(1) = f(1)$ and $f^+(t_d-z_d+2) = f(t_d-z_d+2)$. Then, for any $i \in \llbracket 1,t_d-z_d \rrbracket$ such that $f(i+1) = f^+(i+1)$, we have 
	\begin{align}
		& [ 2f(i+1)- f(i)- f(i+2) ]^+ \nonumber \\
		& \geq  2f(i+1)- f(i)- f(i+2)  \nonumber \\ \nonumber
		&\stackrel{(a)} \geq 2f(i+1)- f^+(i)- f^+(i+2)\\
		& \stackrel{(b)}= 2f^+(i+1)- f^+(i)- f^+(i+2) , \label{eqm1}
	\end{align}
		where $(a)$ holds because $f^+ \geq f$, $(b)$ holds because $f(i+1) = f^+(i+1)$. Moreover, for any $i \in \llbracket 1,t_d-z_d \rrbracket$ such that $f(i+1) \neq f^+(i+1)$, we have 
	\begin{align}
	     & [ 2f(i+1)- f(i)- f(i+2) ]^+  \nonumber \\
	     &\geq  0 \nonumber \\
		 &= 2f^+(i+1)- f^+(i)- f^+(i+2), \label{eqm2}
	\end{align}
	where the last equality holds because $f^+$ is linear between $i$ and $i+2$, i.e., $f^+(i+1)- f^+(i)  = f^+(i+2) - f^+(i+1)$. Indeed, by contradiction, assume that  $f^+$ is not linear between $i$ and $i+2$, then, since $f^+$ is concave, we must have  
	\begin{align}
	f^+(i+1) > \frac{f^+(i+2)+ f^+(i)}{2}. \label{eqcontrad}
\end{align}	
	 Next, we have a contradiction by constructing $h_i$, a concave function such that $f \leq h_i < f^+$, as follows: 
	\begin{align*}
	h_i :j \mapsto \begin{cases}
	  f^+(j) & \text{if } j \neq i+1\\
	  \max \left( \frac{f^+(i+2)+ f^+(i)}{2} , f(i+1) \right) &\text{if } j =i+1
	 \end{cases}.
	\end{align*}
We have $f \leq h_i$ (since $f \leq f^+$), and 	$h_i < f^+$ by \eqref{eqcontrad} and because $f^+(i+1) > f(i+1)$ (since $f^+ \geq f$ and $f^+(i+1) \neq f(i+1)$). Then, to show concavity of $h_i$, it is sufficient to show that  $h_i^{\Delta}$ is non-increasing where $h_i^{\Delta}$ is defined as 
\begin{align*}
h_i^{\Delta}:\llbracket 1, t_d-z_d+1  \rrbracket &\to \mathbb{R}\\
	 j & \mapsto h_i(j+1) - h_i(j).
\end{align*}
For $j \in \llbracket 1 , i -2 \rrbracket \cup \llbracket i+2, t_d-z_d+1 \rrbracket $, we have 
\begin{align} \label{eqmonotone1}
h_i^{\Delta}(j+1) \leq h_i^{\Delta}(j) 
\end{align} by definition of  $h_i^{\Delta}$ and concavity of $f^+$. Then, we have
\begin{align*}
h_i^{\Delta}(i) 
& \stackrel{(a)} = h_i(i+1) - h_i(i)  \\
&\stackrel{(b)} = h_i(i+1) -f^+(i) \\
& \stackrel{(c)}\leq f^+(i+1) -f^+(i) \\
& \stackrel{(d)}\leq f^+(i) - f^+(i-1) \\
& \stackrel{(e)}= h_i(i) - h_i(i-1) \\
& \stackrel{(f)}= h_i^{\Delta}(i-1),
\end{align*}
where $(a)$ and $(f)$ hold by definition of $h_i^{\Delta}$, $(b)$ and $(e)$ hold by definition of $h_i$, $(c)$ holds because $h_i < f^+$, $(d)$ holds by concavity of $f^+$. Then, we have
\begin{align*}
h_i^{\Delta}(i+1) 
&\stackrel{(a)}= h_i(i+2) - h_i(i+1)  \\
&\stackrel{(b)}= f^+(i+2) - h_i(i+1)  \\
& \stackrel{(c)} \leq h_i(i+1) - f^+(i) \\
&\stackrel{(d)} = h_i(i+1) - h_i(i) \\
&\stackrel{(e)} = h_i^{\Delta}(i), \numberthis \label{eqmonotone2}
\end{align*}
where $(a)$ and $(e)$ hold by definition of $h_i^{\Delta}$, $(b)$ and $(d)$ hold by definition of $h_i$, $(c)$ holds because $\frac{f^+(i+2)+ f^+(i)}{2} \leq h_i (i+1)$. Then, we also have
\begin{align*}
h_i^{\Delta}(i+2) 
& \stackrel{(a)}= h_i(i+3) - h_i(i+2)  \\
&\stackrel{(b)}= f^+(i+3) - f^+(i+2)  \\
&\stackrel{(c)} \leq f^+(i+2) - f^+(i+1)  \\
&\stackrel{(d)} \leq f^+(i+2) - h_i(i+1) \\
&\stackrel{(e)} = h_i(i+2) - h_i(i+1) \\
&\stackrel{(f)} = h_i^{\Delta}(i+1), \numberthis \label{eqmonotone3}
\end{align*}
where $(a)$ and $(f)$ hold by definition of $h_i^{\Delta}$, $(b)$ and $(e)$ hold by definition of $h_i$, $(c)$ holds by concavity of $f^+$, $(d)$ holds because $h_i < f^+$. Hence, by \eqref{eqmonotone1}, \eqref{eqmonotone2}, and \eqref{eqmonotone3}, $h_i^{\Delta}$ is non-increasing and we have thus proved \eqref{eqm2} by contradiction.
	
	Next, we have 
	\begin{align}
		& \sum_{i=1}^{t_d-z_d}  [ 2 f(i+1)- f(i)  - f(i+2)  ]^+ \nonumber\\ \nonumber
		& \stackrel{(a)}  \geq 		\sum_{i=1}^{t_d-z_d}  [ 2 f^+(i+1)- f^+(i)  - f^+(i+2)  ]\\\nonumber
		& =	\sum_{i=1}^{t_d-z_d}  [ (f^+(i+1)- f^+(i))  - (f^+(i+2) -f^+(i+1))  ]\\\nonumber
 		& = f^+(2)- f^+(1) + f^+(t_d-z_d+2)- f^+(t_d-z_d+1) \\ \nonumber
 		 		& \stackrel{(b)}= f^+(2) \\ 
 		 		 		 		& \stackrel{(c)} \geq  \frac{H(F_d)}{t_d-z_d}  , \nonumber
	\end{align}
	where $(a)$ holds by \eqref{eqm1} and \eqref{eqm2}, $(b)$ holds because $f^+(t_d-z_d+2)= f^+(t_d-z_d+1) = f(t_d-z_d+1)= H(F_d)$ and $f^+(1)=0$, $(c)$ holds  because 
	$f^+(2) = f^+(2) - f^+(1) \geq (f^+(t_d-z_d+1) - f^+(1))/(t_d-z_d)$ by concavity of $f^+$ and where we have used that $f^+(t_d-z_d+1) =H(F_d)$ and $f^+(1)=f(1)= 0$.
\end{proof}

		Next, by  combining Lemmas \ref{lem444} and \ref{lem555}, we have
	\begin{align}
{r}^{(M)}_{d,l}
& \geq H(M_{d,l}) \nonumber   \\
& \stackrel{(a)}\geq H(F_d) \frac{1}{t_d-z_d} \nonumber \\
&\stackrel{(b)} = {r}^{(F)}_{d} \frac{1}{t_d-z_d}, \label{eqindivratepub}
	\end{align}
	where $(a)$ holds by \eqref{eqmin} and \eqref{eqf1}, which is valid for any $f \in \mathcal{F}$, $(b)$ holds by uniformity of $F_d$.

Finally, since \eqref{eqindivratepub} is valid for any private file storage strategy, \eqref{eqindivratepub} is also valid for a  $\left( \left(2^{\tilde r^{(F)}_{d}}\right)_{d \in \mathcal{D}},\left(2^{\tilde r^{(R)}_{d}}\right)_{d \in \mathcal{D}},\left(2^{\tilde r^{(M)}_{d,l,}}\right)_{d \in \mathcal{D},l \in \mathcal{L}},\left(2^{\tilde r^{(S)}_{l}}\right)_{l \in \mathcal{L}}\right)$ file storage strategy that $(\mathbf{t},\mathbf{z})$-achieves $\left(\tilde{r}^{(F)}_{d}\right)_{d \in \mathcal{D}} $, where $\left({\tilde{r}^{(F)}_{d}}\right)_{d \in \mathcal{D}}  \in \mathcal{C}_{F}(\mathbf{t},\mathbf{z})$ is such that $\tilde{r}^{(F)}_{d} = {r}^{(F)}_{d,\star}(\mathbf{t},\mathbf{z})$ and $  \tilde{r}^{(M)}_{d,l}=r^{(M)}_{d,l,\star}(\mathbf{t},\mathbf{z}) $.

 \section{Proof of Equation \eqref{th52}} \label{App_th5}
Consider an arbitrary  $\left( \left(2^{r^{(F)}_d}\right)_{d \in \mathcal{D}}, \left(2^{r^{(R)}_d}\right)_{d \in \mathcal{D}},\right.$ $\left.\left(2^{r^{(M)}_{d,l}}\right)_{d \in \mathcal{D},l \in \mathcal{L}},\left(2^{r^{(S)}_l}\right)_{l \in \mathcal{L}}\right)$ private file storage strategy that $(\mathbf{t},\mathbf{z})$-achieves $\left({r^{(F)}_d}\right)_{d \in \mathcal{D}}$. Fix $d\in\mathcal{D}$.  Let $\mathcal{V} \subseteq \mathcal{L}$ such that $v \triangleq |\mathcal{V}|< z_d$. For $\mathcal{T} \subseteq \mathcal{L} \backslash \mathcal{V}$ and $\mathcal{S} \subseteq \mathcal{L} \backslash (\mathcal{T} \cup \mathcal{V}) $ such that $|\mathcal{T}|= z_d - v$ and $|\mathcal{S}|=t_d-z_d$. Using Definition~\ref{definition_modelg2} and the reliability and security constraints \eqref{eqrel2} and~\eqref{eqSeca2} from Definition~\ref{def2b}, we first derive the following lemma. 

\begin{lem} \label{lem777}
We have 
 \begin{align*}
   \sum_{l \in \mathcal{S}} H(M_{d,l},K_{d,l} |M_{d,\mathcal{V}},K_{d,\mathcal{V}} )  \geq    H( F_{d} ) + \sum_{l \in \mathcal{S}}H(K_{d,l}) . \numberthis \label{eqRanDint0}
 \end{align*}
\end{lem}
\begin{proof}
We have 
 \begin{align*}
 &  \sum_{l \in \mathcal{S}} H(M_{d,l},K_{d,l} |M_{d,\mathcal{V}},K_{d,\mathcal{V}} ) \\
 &\stackrel{(a)}  \geq   H(M_{d,\mathcal{S}},K_{d,\mathcal{S}}|M_{d,\mathcal{V}},K_{d,\mathcal{V}} )\\
 & \stackrel{(b)}\geq H(M_{d,\mathcal{S}},K_{d,\mathcal{S}} |M_{d,\mathcal{V}\cup \mathcal{T}},K_{d,\mathcal{V}\cup \mathcal{T}})   \\
 & =    I(M_{d,\mathcal{S}},K_{d,\mathcal{S}}; F_{d} | M_{d,\mathcal{V}\cup \mathcal{T}},K_{d,\mathcal{V}\cup \mathcal{T}} ) \\
 & \phantom{--} + H(M_{d,\mathcal{S}},K_{d,\mathcal{S}} | F_{d} ,M_{d,\mathcal{V}\cup \mathcal{T}},K_{d,\mathcal{V}\cup \mathcal{T}}) \\
 & =   H( F_{d} | M_{d,\mathcal{V}\cup \mathcal{T}},K_{d,\mathcal{V}\cup \mathcal{T}})\\
 & \phantom{--}- H(F_{d} |M_{d,\mathcal{V}\cup \mathcal{T}\cup \mathcal{S}},K_{d,\mathcal{V}\cup \mathcal{T}\cup \mathcal{S}}) \\
 & \phantom{--}+ H(M_{d,\mathcal{S}},K_{d,\mathcal{S}} | F_{d} ,M_{d,\mathcal{V}\cup \mathcal{T}},K_{d,\mathcal{V}\cup \mathcal{T}}) \\
  & \stackrel{(c)} \geq    H( F_{d} | M_{d,\mathcal{V}\cup \mathcal{T}},K_{d,\mathcal{V}\cup \mathcal{T}})- H(F_{d} |S_{\mathcal{V}\cup \mathcal{T}\cup \mathcal{S},d} ) \\
 & \phantom{--}+ H(M_{d,\mathcal{S}},K_{d,\mathcal{S}} | F_{d} ,M_{d,\mathcal{V}\cup \mathcal{T}},K_{d,\mathcal{V}\cup \mathcal{T}}) \\
    &  \stackrel{(d)} \geq    H( F_{d} | M_{d,\mathcal{V}\cup \mathcal{T}},K_{d,\mathcal{V}\cup \mathcal{T}})- H(F_{d} |\widehat F_d(\mathcal{V}\cup \mathcal{T}\cup \mathcal{S}) ) \\
 & \phantom{--}+ H(M_{d,\mathcal{S}},K_{d,\mathcal{S}} | F_{d} ,M_{d,\mathcal{V}\cup \mathcal{T}},K_{d,\mathcal{V}\cup \mathcal{T}}) \\
    & \stackrel{(e)} =  H( F_{d}  ) + H(M_{d,\mathcal{S}},K_{d,\mathcal{S}} | F_{d} ,M_{d,\mathcal{V}\cup \mathcal{T}},K_{d,\mathcal{V}\cup \mathcal{T}})\\
        & \stackrel{(f)} \geq  H( F_{d} ) + H(M_{d,\mathcal{S}},K_{d,\mathcal{S}} | F_{d} ,R_{d},K_{d,\mathcal{V} \cup \mathcal T}) \\
                & \stackrel{(g)}  \geq  H( F_{d} ) + H(K_{d,\mathcal{S}} | F_{d} ,R_{d},K_{d,\mathcal{V}\cup \mathcal T}) \\
                                & \stackrel{(h)}  =  H( F_{d} ) + H(K_{d,\mathcal{S}} ) \\
                                &  \stackrel{(i)}=  H( F_{d} ) + \sum_{l \in \mathcal{S}}H(K_{d,l}) , 
 \end{align*}
where $(a)$ holds by the chain rule and because conditioning reduces entropy, $(b)$ holds because conditioning reduces entropy, $(c)$ holds because $S_{\mathcal{V}\cup \mathcal{T}\cup \mathcal{S},d} $ is a function of $( M_{d,\mathcal{V}\cup \mathcal{T}},K_{d,\mathcal{V}\cup \mathcal{T}})$, $(d)$~holds because $\widehat F_d(\mathcal{V}\cup \mathcal{T}\cup \mathcal{S})$ is a function of $S_{\mathcal{V}\cup \mathcal{T}\cup \mathcal{S},d}$, $(e)$ holds by \eqref{eqrel2} because $|\mathcal{V}\cup \mathcal{T}\cup \mathcal{S}| = t_d$ and by \eqref{eqSeca2} because $|\mathcal{V}\cup \mathcal{T}|=z_d$, $(f)$~holds because $M_{d,\mathcal{V}\cup \mathcal{T}}$ is a function of $(F_{d} ,R_{d},K_{d,\mathcal{V}\cup \mathcal T})$, $(g)$ holds by the chain rule and non-negativity of entropy, $(h)$ holds by independence between $K_{d,\mathcal{S}}$ and $(F_{d} ,R_{d},K_{d,\mathcal{V}\cup \mathcal T})$ because $\mathcal{S} \cap (\mathcal{V}\cup \mathcal T) = \emptyset$, $(i)$ holds by independence of the keys $(K_{d,l})_{l\in\mathcal S}$.
\end{proof}

Next, by summing both side of the equation of Lemma \ref{lem777} over of all possible sets $\mathcal{T} \subseteq \mathcal{L} \backslash \mathcal{V}$ and $\mathcal{S} \subseteq \mathcal{L} \backslash (\mathcal{T} \cup \mathcal{V}) $ such that $|\mathcal{T}|= z_d - v$ and $|\mathcal{S}|=t_d-z_d$, we obtain the following lemma.

\begin{lem} \label{lem888}
Consider $$l^\star(\mathcal{V})  \in \argmax_{l\in\mathcal{L} \backslash \mathcal{V}}\left[ H(M_{d,l},K_{d,l} |M_{d,\mathcal{V}},K_{d,\mathcal{V}} )  - H(K_{d,l}) \right].$$
 We have
 \begin{align*} 
&  \frac{1}{t_d-z_d} H( F_{d} )\\
& \leq     H(M_{d,\mathcal{L}},K_{d,\mathcal{L}} |F_d, M_{d,\mathcal{V}},K_{d,\mathcal{V}} ) \\
&\phantom{--}  - H(M_{d,\mathcal{L}},K_{d,\mathcal{L}} |F_d, M_{d,\mathcal{V} \cup \{l^\star(\mathcal{V})\}},K_{d,\mathcal{V}  \cup \{l^\star(\mathcal{V})\}} ) -  n_d  \numberthis \label{eqint0}.
 \end{align*}

\end{lem}
\begin{proof}
We have
   \begin{align*} 
& \frac{1}{t_d-z_d} H( F_{d} )\\
 & \stackrel{(a)} = \Omega_d \sum_{ \substack{\mathcal{T} \subseteq \mathcal{L}\backslash \mathcal{V}  \\  |\mathcal{T} | = z_d -v}} \sum_{ \substack{\mathcal{S} \subseteq  \mathcal{L} \backslash (\mathcal{T} \cup \mathcal{V}) \\  |\mathcal{S} | = t_d-z_d }}  H( F_{d} ) \\ 
   &  \stackrel{(b)} \leq \Omega_d  \sum_{ \substack{\mathcal{T} \subseteq \mathcal{L}\backslash \mathcal{V}  \\  |\mathcal{T} | = z_d -v}} \sum_{ \substack{\mathcal{S} \subseteq  \mathcal{L} \backslash (\mathcal{T} \cup \mathcal{V}) \\  |\mathcal{S} | = t_d-z_d }} \\
   & \phantom{----}\sum_{l \in \mathcal{S}} \left[ H(M_{d,l},K_{d,l} |M_{d,\mathcal{V}},K_{d,\mathcal{V}} )   - H(K_{d,l}) \right] \\
      &  \stackrel{(c)} = \Omega_d \sum_{ \substack{\mathcal{T} \subseteq \mathcal{L}\backslash \mathcal{V}  \\  |\mathcal{T} | = z_d -v}}   { L - z_d  -  1 \choose t_d - z_d - 1 }  \\
 & \phantom{--} \times \!\! \sum_{ l \in  \mathcal{L} \backslash (\mathcal{T} \cup \mathcal{V})  }      \left[ H(M_{d,l},K_{d,l} |M_{d,\mathcal{V}},K_{d,\mathcal{V}} )   -  H(K_{d,l}) \right] \\
            &  \stackrel{(d)} = \Omega_d{ L-z_d - 1 \choose t_d-z_d-1 } \\
 & \phantom{--} \times \!\!\sum_{ \substack{\mathcal{T} \subseteq \mathcal{L} \backslash  \mathcal{V}  \\  |\mathcal{T} | = L -z_d}} \sum_{ l \in  \mathcal{T}   } \left[H(M_{d,l},K_{d,l} |M_{d,\mathcal{V}},K_{d,\mathcal{V}} )  - H(K_{d,l}) \right]  \\
                        &  \stackrel{(e)} = \Omega_d { L-z_d  - 1 \choose t_d-z_d-1 } {L-v -1\choose L-z_d-1} \\
 & \phantom{--} \times \sum_{ l \in  \mathcal{L} \backslash   \mathcal{V}  }  \left[H(M_{d,l},K_{d,l} |M_{d,\mathcal{V}},K_{d,\mathcal{V}} )  - H(K_{d,l}) \right]  \\
      & = \frac{1}{L-v}\sum_{ l \in  \mathcal{L} \backslash   \mathcal{V}  } \left[ H(M_{d,l},K_{d,l} |M_{d,\mathcal{V}},K_{d,\mathcal{V}} )  - H(K_{d,l}) \right] \\
    & \stackrel{(f)} \leq     \left[  H(M_{d,l^\star(\mathcal{V})},K_{d,l^\star(\mathcal{V})} |M_{d,\mathcal{V}},K_{d,\mathcal{V}} )  -  H(K_{d,l^\star(\mathcal{V})})  \right] \\
    & =    H(M_{d,l^\star(\mathcal{V})},K_{d,l^\star(\mathcal{V})} | M_{d,\mathcal{V}},K_{d,\mathcal{V}} )  -  n_d   \\
        & \stackrel{(g)}=    H(M_{d,l^\star(\mathcal{V})},K_{d,l^\star(\mathcal{V})} |F_d, M_{d,\mathcal{V}},K_{d,\mathcal{V}} )  -  n_d   \\
        &  \stackrel{(h)} =  H(M_{d,\mathcal{L}},K_{d,\mathcal{L}} |F_d, M_{d,\mathcal{V}},K_{d,\mathcal{V}} )  \\
&\phantom{--}- H(M_{d,\mathcal{L}},K_{d,\mathcal{L}} |F_d, M_{d,\mathcal{V} \cup \{l^\star(\mathcal{V})\}},K_{d,\mathcal{V}  \cup \{l^\star(\mathcal{V})\}} ) -  n_d ,
 \end{align*}
 where $(a)$ holds with $\Omega_d \triangleq \frac{1}{t_d-z_d} {L-v \choose z_d-v}^{-1} {{L-z_d}\choose{t_d-z_d}}^{-1} $, $(b)$~holds by \eqref{eqRanDint0}, $(c)$ holds because for any $l\in\mathcal{L} \backslash (\mathcal{T} \cup \mathcal{V})$, the term $ \left[ H(M_{d,l},K_{d,l} |M_{d,\mathcal{V}},K_{d,\mathcal{V}} )  - H(K_{d,l}) \right]$ appears exactly ${ L-z_d - 1 \choose t_d-z_d-1 }$ times in the term $\sum_{ \substack{\mathcal{S} \subseteq  \mathcal{L} \backslash (\mathcal{T} \cup \mathcal{V}) \\  |\mathcal{S} | = t_d-z_d }} \sum_{l \in \mathcal{S}} \left[ H(M_{d,l},K_{d,l} |M_{d,\mathcal{V}},K_{d,\mathcal{V}} )  - H(K_{d,l}) \right]$ (note that a similar argument is made in \cite[Lemma 3.2]{de1999multiple}), $(d)$ holds by a change of variables in the sums, $(e)$ holds because for any $l\in\mathcal{L}\backslash \mathcal{V} $, $ \left[ H(M_{d,l},K_{d,l} |M_{d,\mathcal{V}},K_{d,\mathcal{V}} )  - H(K_{d,l}) \right]$ appears exactly ${ L - v-1 \choose L-z_d-1 }$ times in the term $\sum_{ \substack{\mathcal{T} \subseteq \mathcal{L} \backslash \mathcal{V}  \\  |\mathcal{T} | =L- z }}   \sum_{l \in   \mathcal{T}} \left[ H(M_{d,l},K_{d,l} |M_{d,\mathcal{V}},K_{d,\mathcal{V}} )  - H(K_{d,l}) \right]$,  $(f)$~ holds by the definition of $l^\star(\mathcal{V})  $, $(g)$ holds because $I(F_d; M_{d,\mathcal{V} \cup \{l^\star(\mathcal{V})\}},K_{d,\mathcal{V}  \cup \{l^\star(\mathcal{V})\}})=0$ by \eqref{eqSeca} and since $|\mathcal{V}  \cup \{l^\star(\mathcal{V})\}| \leq z_d$, $(h)$ holds by the chain~rule.
\end{proof}

Next, we apply multiple times Lemma \ref{lem888} to obtain the following lemma. 

\begin{lem} \label{lem999}
Define $\mathcal{V}_0 \triangleq \emptyset$ and for $i\in \llbracket 1, z_d \rrbracket$, $\mathcal{V}_i \triangleq \mathcal{V}_{i-1} \cup \{ l^{\star} (\mathcal{V}_{i-1}) \}$.  We have
\begin{align*} 
       \frac{z_d}{t_d-z_d} H( F_{d} ) 
       & \leq  H(M_{d,\mathcal{L}},K_{d,\mathcal{L}} |F_d  ) \\
       &\!- H(M_{d,\mathcal{L}},K_{d,\mathcal{L}} |F_d, M_{d,\mathcal{V}_{z_d}},K_{d,\mathcal{V}_{z_d}} ) - n_dz_d.
  \end{align*}
 \end{lem}
\begin{proof}
We have
\begin{align*} 
  & \frac{z_d}{t_d-z_d} H( F_{d} )\\
 & =\smash\sum_{i=0}^{z_d-1} \frac{1}{t_d-z_d} H( F_{d} )\\
        & \leq \smash\sum_{i=0}^{z_d-1}[ H(M_{d,\mathcal{L}},K_{d,\mathcal{L}} |F_d, M_{d,\mathcal{V}_i},K_{d,\mathcal{V}_i} ) \\
       & \phantom{--} - H(M_{d,\mathcal{L}},K_{d,\mathcal{L}} |F_d, M_{d,\mathcal{V}_{i+1}},K_{d,\mathcal{V}_{i+1}} ) - n_d]   \\
                & =  H(M_{d,\mathcal{L}},K_{d,\mathcal{L}} |F_d  ) \\
       & \phantom{--} - H(M_{d,\mathcal{L}},K_{d,\mathcal{L}} |F_d, M_{d,\mathcal{V}_{z_d}},K_{d,\mathcal{V}_{z_d}} ) - n_dz_d,
 \end{align*}
  where the inequality holds by applying $z_d$ times Equation \eqref{eqint0}  and the definition of $\mathcal{V}_i$, $i \in \llbracket 0, z_d \rrbracket$.

\end{proof} 
Finally, we simplify the upper bound of Lemma \ref{lem999} as follows. We have
 \begin{align*} 
   &   \frac{z_d}{t_d-z_d} {r}^{(F)}_{d}\\
  & \stackrel{(a)} =  \frac{z_d}{t_d-z_d} H( F_{d} )\\
 & \stackrel{(b)}\leq  H(M_{d,\mathcal{L}},K_{d,\mathcal{L}} |F_d  )  \\
       & \phantom{--}- H(M_{d,\mathcal{L}},K_{d,\mathcal{L}} |F_d, M_{d,\mathcal{V}_{z_d}},K_{d,\mathcal{V}_{z_d}} ) - n_dz_d\\
                    & \stackrel{(c)} \leq  H(M_{d,\mathcal{L}},K_{d,\mathcal{L}} |F_d  ) \\
       & \phantom{--} - H(K_{d,\mathcal{L}} |F_d, M_{d,\mathcal{V}_{z_d}},K_{d,\mathcal{V}_{z_d}} ) - n_d z_d   \\
                                        & \stackrel{(d)}\leq  H(M_{d,\mathcal{L}},K_{d,\mathcal{L}} |F_d  )  - H(K_{d,\mathcal{L}} |F_d, R_{d},K_{d,\mathcal{V}_{z_d}} ) - n_d z_d   \\                                                                             
& =  H(M_{d,\mathcal{L}},K_{d,\mathcal{L}} |F_d  )  - H(K_{d,\mathcal{V}_{z_d}^c} |F_d, R_{d} ) - n_d z_d   \\                                                                                                                        
                                         & \stackrel{(e)} =  H(M_{d,\mathcal{L}},K_{d,\mathcal{L}} |F_d  )  - H(K_{d,\mathcal{V}_{z_d}^c}  ) - n_dz_d   \\                                                                                 & \stackrel{(f)} =  H(M_{d,\mathcal{L}},K_{d,\mathcal{L}} |F_d  )  - n_d L    \\
                                        & \stackrel{(g)} \leq   H(R_d, F_d,K_{d,\mathcal{L}} |F_d  )  - n_dL   \\                                                                              
                                        & =   H(R_d,K_{d,\mathcal{L}}|F_d)  - n_d L   \\ 
                                        & \stackrel{(h)} =   H(R_d) + H(K_{d,\mathcal{L}})  - n_d L   \\
                                        & \stackrel{(i)}=   H(R_d)   \\
                                        & \stackrel{(j)}=  r^{(R)}_d   , \numberthis \label{eqratelrand}
 \end{align*}
 where $(a)$ holds by uniformity of $F_d$, $(b)$ holds by Lemma~\ref{lem999}, $(c)$ holds by the chain rule and non-negativity of entropy, $(d)$~holds because $M_{d,\mathcal{V}_{z_d}}$ is a function of $(F_{d} ,R_{d},K_{d,\mathcal{V}_{z_d}})$, $(e)$ holds by independence between $K_{d,\mathcal{V}_{z_d}^c}$ and $(F_d, R_{d} )$, $(f)$~holds by the uniformity of the keys $(K_{d,l})_{l\in \mathcal{V}_{z_d}^c}$, $(g)$~holds because $M_{d,\mathcal{L}},K_{d,\mathcal{L}}$ is a function of $(R_d, F_d,K_{d,\mathcal{L}})$, $(h)$~holds by mutual independence between $R_d$, $F_d$, and $K_{d,\mathcal{L}}$, $(i)$~holds by the uniformity of the keys $(K_{d,l})_{l\in \mathcal{L}}$, $(j)$ holds by uniformity of $R_d$.
 
 Since~\eqref{eqratelrand} is valid for any private file storage strategy, \eqref{eqratelrand} is also valid for a  $\left( \left(2^{\tilde r^{(F)}_{d}}\right)_{d \in \mathcal{D}},\left(2^{\tilde r^{(R)}_{d}}\right)_{d \in \mathcal{D}},\left(2^{\tilde r^{(M)}_{d,l,}}\right)_{d \in \mathcal{D},l \in \mathcal{L}},\left(2^{\tilde r^{(S)}_{l}}\right)_{l \in \mathcal{L}}\right)$ file storage strategy that $(\mathbf{t},\mathbf{z})$-achieves $\left(\tilde{r}^{(F)}_{d}\right)_{d \in \mathcal{D}} $, where $\left({\tilde{r}^{(F)}_{d}}\right)_{d \in \mathcal{D}}  \in \mathcal{C}_{F}(\mathbf{t},\mathbf{z})$ is such that $\tilde{r}^{(F)}_{d} = {r}^{(F)}_{d,\star}(\mathbf{t},\mathbf{z})$ and $  \tilde{r}^{(R)}_{d}=r^{(R)}_{d,\star}(\mathbf{t},\mathbf{z})$.

 \section{Proof of Theorem \ref{th72}} \label{App_th7}
We first review the notion of ramp secret sharing \cite{yamamoto1986secret,blakley1984security}  in Section~\ref{seca}. We then present our achievability scheme and its analysis in Section \ref{secb}.

\subsection{Review of ramp secret sharing} \label{seca}

\begin{defn}[\!\!{\cite{yamamoto1986secret,blakley1984security}}]
Let $t \in \llbracket 1 , L \rrbracket$ and $z \in \llbracket 1,t-1 \rrbracket$. A $(t,z,L)$-ramp secret sharing scheme consists~of 
	\begin{itemize}
	\item A secret $S$ uniformly distributed over $\{0,1\}^{n_s}$;
		\item A stochastic encoder 
$ e : \{ 0,1\}^{n_s} \times \{ 0,1\}^{n_r} 
\to \{ 0,1\}^{n_{sh}L},  
(S,R) \mapsto  (H_l)_{l \in \mathcal{L}},$ which takes as input the secret $S$ and a randomization sequence $R$ uniformly distributed over $\{ 0,1\}^{n_r}$ and independent of $S$, and outputs $L$ shares $(H_l)_{l \in \mathcal{L}}$ of length $n_{sh}$. For any $\mathcal{S} \subseteq \mathcal{L}$, we define $H_{\mathcal{S}} \triangleq (H_l)_{l\in\mathcal{S}}$;
\end{itemize}
and satisfies the two conditions
\begin{align}
\displaystyle\max_{\substack{\mathcal{T}\subseteq \mathcal{L}:|\mathcal{T}|=t}}	H(S|H_{\mathcal{T}}) & = 0, \text{ (Recoverability)} \label{eqreq1} \\
	 \displaystyle\max_{\substack{\mathcal{U}\subseteq \mathcal{L}:|\mathcal{U}|\leq z}} I(S;H_{\mathcal{U}})  &=0. \text{ (Security)} \label{eqreq2} 
	 \end{align}
\end{defn}

\begin{thm}[\!\!{\cite{yamamoto1986secret,blakley1984security}}] \label{th8}
Let $t \in \llbracket 1 , L \rrbracket$ and $z \in \llbracket 1,t-1 \rrbracket$. For a fixed secret length $n_s$, there exists a  $(t,z,L)$-ramp secret sharing scheme such that the length of a share $n_{sh}$ and the length of the randomization sequence $n_r$ satisfy
\begin{align*}
n_{sh}= \frac{n_s}{t-z},\quad
n_r = \frac{n_sz}{t-z}.
\end{align*}
\end{thm}

\subsection{Achievability scheme for Theorem \ref{th72}}\label{secb}
\emph{Coding scheme}: Fix $d \in \mathcal{D}$ and consider a file $F_d$ such that $r_d^{(F)}=|F_d| =n_d(t_d -z_d)$. Then, User~$d$ forms $(H_{l,d})_{l\in\mathcal{L}}$ with a $(t_d,z_d,L)$-ramp secret sharing scheme taken from Theorem~\ref{th8} applied to $F_d$. By Theorem~\ref{th8}, for $l\in \mathcal{L}$,  the length of a share is $|H_{l,d}|=\frac{|F_d|}{t_d-z_d} = n_d$, and the length of the randomization sequence is $n_r = \frac{|F_d|z_d}{t_d-z_d} =n_d z_d$. Hence, since $|K_{d,l}|=n_d$, $l\in \mathcal{L}$, User~$d$ can form $M_{d,l} \triangleq H_{l,d} \oplus K_{d,l}$ and publicly send it to Server $l$, where $\oplus$ denotes bitwise modulo-two addition. Upon receiving $M_{d,l}$, Server $l$ stores $S_{l,d} \triangleq K_{d,l} \oplus M_{d,l} = H_{l,d}$.

\emph{Resources needed}: For $d \in \mathcal{D}$, $l \in \mathcal{L}$, the length of the randomization sequence at User~$d$ is $r_d^{(R)} = n_r = n_d z_d$, the length of the public communication from User $d$ to Server $l$ is $r^{(M)}_{d,l} = |M_{d,l}| = n_d$, and the storage needed at Server $l$ is $r^{(S)}_l = \sum_{d\in\mathcal{D}}|M_{d,l}|  = \sum_{d\in\mathcal{D}}n_d$.

\emph{Analysis of recoverability}: For $d \in \mathcal{D}$, consider an arbitrary subset $\mathcal{A} \subseteq \mathcal{L}$ of $t_d$ servers that pool their information, they then have access to $(S_{l,d})_{l\in\mathcal{A}} = (H_{l,d})_{l\in\mathcal{A}}$ such that by Theorem \ref{th8}, $H(F_d | (H_{l,d})_{l\in\mathcal{A}})=0$ since $|\mathcal{A}|=t_d$.

\emph{Analysis of security}: For $d \in \mathcal{D}$, consider an arbitrary subset $\mathcal{U} \subseteq \mathcal{L}$ of $z_d$ colluding servers. Then, we have
\begin{align*}
&I(F_{\mathcal{Z}_d}; M_{\mathcal{D}},K_{\mathcal{D},\mathcal{U}})\\
& \stackrel{(a)} = I(F_{\mathcal{Z}_d}; M_{\mathcal{Z}_d},K_{\mathcal{Z}_d,\mathcal{U}}) \\
       & \phantom{--}+ I(F_{\mathcal{Z}_d}; M_{\mathcal{Z}_d^c},K_{\mathcal{Z}_d^c,\mathcal{U}}|M_{\mathcal{Z}_d},K_{\mathcal{Z}_d,\mathcal{U}}) \\
& \stackrel{(b)} = I(F_{\mathcal{Z}_d}; M_{\mathcal{Z}_d},K_{\mathcal{Z}_d,\mathcal{U}})  \\
& \stackrel{(c)} = I(F_{\mathcal{Z}_d\backslash\{d\}}; M_{\mathcal{Z}_d\backslash\{d\}},K_{\mathcal{Z}_d\backslash\{d\},\mathcal{U}})  \\
       & \phantom{--}+ I(F_{\mathcal{Z}_d\backslash\{d\}}; M_{d},K_{d,\mathcal{U}}| M_{\mathcal{Z}_d\backslash\{d\}},K_{\mathcal{Z}_d\backslash\{d\},\mathcal{U}}) \\
& \phantom{--}  + I(F_{d}; M_{\mathcal{Z}_d},K_{\mathcal{Z}_d,\mathcal{U}}|F_{\mathcal{Z}_d\backslash\{d\}})\\
& \stackrel{(d)} = I(F_{\mathcal{Z}_d\backslash\{d\}}; M_{\mathcal{Z}_d\backslash\{d\}},K_{\mathcal{Z}_d\backslash\{d\},\mathcal{U}})   \\
       & \phantom{--}+ I(F_{d}; M_{\mathcal{Z}_d},K_{\mathcal{Z}_d,\mathcal{U}}|F_{\mathcal{Z}_d\backslash\{d\}})\\
& \leq I(F_{\mathcal{Z}_d\backslash\{d\}}; M_{\mathcal{Z}_d\backslash\{d\}},K_{\mathcal{Z}_d\backslash\{d\},\mathcal{U}})   \\
       & \phantom{--}+ I(F_{d}; M_{\mathcal{Z}_d},K_{\mathcal{Z}_d,\mathcal{U}},F_{\mathcal{Z}_d\backslash\{d\}})\\
& = I(F_{\mathcal{Z}_d\backslash\{d\}}; M_{\mathcal{Z}_d\backslash\{d\}},K_{\mathcal{Z}_d\backslash\{d\},\mathcal{U}})   + I(F_{d}; M_{d },K_{d ,\mathcal{U}} ) \\
       & \phantom{--}+ I(F_{d}; M_{\mathcal{Z}_d\backslash\{d\}},K_{\mathcal{Z}_d\backslash\{d\},\mathcal{U}},F_{\mathcal{Z}_d\backslash\{d\}}|M_{d },K_{d ,\mathcal{U}} )\\
& \stackrel{(e)}= I(F_{\mathcal{Z}_d\backslash\{d\}}; M_{\mathcal{Z}_d\backslash\{d\}},K_{\mathcal{Z}_d\backslash\{d\},\mathcal{U}})   + I(F_{d}; M_{d },K_{d ,\mathcal{U}} )\\
& \stackrel{(f)}\leq \sum_{i \in \mathcal{Z}_d} I(F_{i}; M_{i},K_{i,\mathcal{U}} ) \\
& \stackrel{(g)}= \sum_{i \in \mathcal{Z}_d} I(F_{i}; H_{\mathcal{U},i},M_{i,\mathcal{U}^c},K_{i,\mathcal{U}} ) \\
& = \sum_{i \in \mathcal{Z}_d} [I(F_{i}; H_{\mathcal{U},i}) +I(F_{i}; M_{i,\mathcal{U}^c},K_{i,\mathcal{U}}|H_{\mathcal{U},i} )]\\
& \stackrel{(h)}= \sum_{i \in \mathcal{Z}_d} I(F_{i}; M_{i,\mathcal{U}^c},K_{i,\mathcal{U}}|H_{\mathcal{U},i} ) \\
& = \sum_{i \in \mathcal{Z}_d} [I(F_{i}; K_{i,\mathcal{U}}|H_{\mathcal{U},i} )+I(F_{i}; M_{i,\mathcal{U}^c}|H_{\mathcal{U},i},K_{i,\mathcal{U}} )] \\
& \stackrel{(j)}= \sum_{i \in \mathcal{Z}_d} I(F_{i}; M_{i,\mathcal{U}^c}|H_{\mathcal{U},i},K_{i,\mathcal{U}} ) \\
& \leq \sum_{i \in \mathcal{Z}_d} I(F_{i},H_{\mathcal{U}^c,i}; M_{i,\mathcal{U}^c}|H_{\mathcal{U},i},K_{i,\mathcal{U}} ) \\
& = \sum_{i \in \mathcal{Z}_d} [H( M_{i,\mathcal{U}^c}|H_{\mathcal{U},i},K_{i,\mathcal{U}} ) \\
       & \phantom{-----}- H(M_{i,\mathcal{U}^c}|F_{i},H_{\mathcal{U}^c,i},H_{\mathcal{U},i},K_{i,\mathcal{U}}) ]\\
& \leq \sum_{i \in \mathcal{Z}_d} [ | M_{i,\mathcal{U}^c}| - H(M_{i,\mathcal{U}^c}|F_{i},H_{\mathcal{U}^c,i},H_{\mathcal{U},i},K_{i,\mathcal{U}}) ]\\
& \stackrel{(k)}= \sum_{i \in \mathcal{Z}_d} [ | M_{i,\mathcal{U}^c}| - H(K_{i,\mathcal{U}^c}|F_{i},H_{\mathcal{U}^c,i},H_{\mathcal{U},i},K_{i,\mathcal{U}}) ]\\
& \stackrel{(l)}= \sum_{i \in \mathcal{Z}_d} [ | M_{i,\mathcal{U}^c}| - H(K_{i,\mathcal{U}^c})  ]\\
&\stackrel{(m)} =0,
\end{align*}
where $(a)$ holds by the chain rule and $\mathcal{Z}_d^c$ denotes the complement of $\mathcal{Z}_d$ in $\mathcal{D}$, $(b)$ holds because $I(F_{\mathcal{Z}_d}; M_{\mathcal{Z}_d^c},K_{\mathcal{Z}_d^c,\mathcal{U}}|M_{\mathcal{Z}_d},K_{\mathcal{Z}_d,\mathcal{U}}) \leq I(F_{\mathcal{Z}_d},M_{\mathcal{Z}_d},K_{\mathcal{Z}_d,\mathcal{U}}; M_{\mathcal{Z}_d^c},K_{\mathcal{Z}_d^c,\mathcal{U}}) = 0$, $(c)$ holds by the chain rule applied twice, $(d)$ holds because $I(F_{\mathcal{Z}_d\backslash\{d\}}; M_{d},K_{d,\mathcal{U}}| M_{\mathcal{Z}_d\backslash\{d\}},K_{\mathcal{Z}_d\backslash\{d\},\mathcal{U}}) \leq I(F_{\mathcal{Z}_d\backslash\{d\}}, M_{\mathcal{Z}_d\backslash\{d\}},K_{\mathcal{Z}_d\backslash\{d\},\mathcal{U}}; M_{d},K_{d,\mathcal{U}})=0 $, $(e)$ holds because $I(F_{d}; M_{\mathcal{Z}_d\backslash\{d\}},K_{\mathcal{Z}_d\backslash\{d\},\mathcal{U}},F_{\mathcal{Z}_d\backslash\{d\}}|M_{d },K_{d ,\mathcal{U}} ) \leq I(F_{d},M_{d },K_{d ,\mathcal{U}} ; M_{\mathcal{Z}_d\backslash\{d\}},K_{\mathcal{Z}_d\backslash\{d\},\mathcal{U}},F_{\mathcal{Z}_d\backslash\{d\}}) = 0$, $(f)$ holds by iterating the steps between $(b)$ and $(e)$, $(g)$ holds because $M_{i} = (M_{i,\mathcal{U}^c},M_{i,\mathcal{U}}) = (M_{i,\mathcal{U}^c},(H_{l,i} \oplus K_{i,l})_{l\in\mathcal{U}}) $, $(h)$~holds because $I(F_{i}; H_{\mathcal{U},i})=0$ by Theorem \ref{th8} and since $|\mathcal{U}|=z_d \leq z_i$ for any $i \in \mathcal{Z}_d$, $(j)$~holds because $I(F_{i}; K_{i,\mathcal{U}}|H_{\mathcal{U},i} ) \leq I(F_{i},H_{\mathcal{U},i} ; K_{i,\mathcal{U}}) = 0$, $(k)$ holds because $ M_{i,\mathcal{U}^c}  =  (H_{l,i} \oplus K_{i,l})_{l\in\mathcal{U}^c} $, $(l)$ holds by independence between $K_{i,\mathcal{U}^c}$ and $(F_{i},H_{\mathcal{U}^c,i},H_{\mathcal{U},i},K_{i,\mathcal{U}})$, $(m)$ holds because by uniformity of $K_{i,\mathcal{U}^c}$, $H(K_{i,\mathcal{U}^c}) = |K_{i,\mathcal{U}^c}| = |\mathcal{U}^c|n_i =| M_{i,\mathcal{U}^c}| $ for any $i \in \mathcal{Z}_d$. 
 \section{Proof of Theorem \ref{th62}} \label{App_th6}
The achievability scheme presented in the proof of Theorem \ref{th72} provides a $\left( \left(2^{r^{(F)}_d}\right)_{d \in \mathcal{D}}, \left(2^{r^{(R)}_d}\right)_{d \in \mathcal{D}},\left(2^{r^{(M)}_{d,l}}\right)_{d \in \mathcal{D},l \in \mathcal{L}},\left(2^{r^{(S)}_l}\right)_{l \in \mathcal{L}}\right)$ private file storage strategy that $(\mathbf{t},\mathbf{z})$-achieves $\left({r^{(F)}_d}\right)_{d \in \mathcal{D}}$ such that for any $d \in \mathcal{D}$
\begin{align}
r^{(F)}_d &= n_d (t_d-z_d), \label{eqachiev1f}\\
r^{(R)}_d &= n_d z_d,  \label{eqachiev1r}\\
r^{(M)}_{d,l} & = n_d ,   \forall l\in \mathcal{L},\label{eqachiev1mi}\\
\textstyle\sum_{l \in\mathcal{L}} r^{(M)}_{d,l} & = L n_d  , \label{eqachiev1rind}\\
r^{(S)}_{l} &= \textstyle\sum_{d \in \mathcal D} n_d   , \forall l\in \mathcal{L}. \label{eqachiev1s}
\end{align}
Next, by \eqref{eqachiev1f} and \eqref{th12}, we have $r^{(F)}_{d, \star} (\mathbf{t},\mathbf{z})= n_d (t_d-z_d), \forall d \in \mathcal{D}.$
By \eqref{eqachiev1r} and \eqref{th52}, we have 
$r^{(R)}_{d,\star} (\mathbf{t},\mathbf{z}) =  n_d z_d = \frac{ z_d }{t_d-z_d} r^{(F)}_{d,\star}(\mathbf{t},\mathbf{z}), \forall d \in \mathcal{D}.$
By~\eqref{eqachiev1rind} and \eqref{th32}, we have 
$r^{(M)}_{d,\Sigma,\star}(\mathbf{t},\mathbf{z})= L n_d = \frac{  L }{t_d-z_d} r^{(F)}_{d,\star}(\mathbf{t},\mathbf{z}), \forall d\in \mathcal{D}.$ 
By~\eqref{eqachiev1s} and \eqref{th22}, we have 
 $r^{(S)}_{d,l,\star}(\mathbf{t},\mathbf{z}) = \sum_{d \in \mathcal D} n_d, \forall d\in \mathcal{D}, \forall l\in \mathcal{L}. $
Assume that \eqref{eqsym2} holds, by \eqref{eqachiev1mi} and \eqref{th42}, we have 
$r^{(M)}_{d,l,\star} (\mathbf{t},\mathbf{z}) = n_d = \frac{ 1 }{t_d-z_d} r^{(F)}_{d,\star}(\mathbf{t},\mathbf{z}), \forall d\in \mathcal{D}, \forall l \in \mathcal{L}. $

\bibliographystyle{IEEEtran}
\bibliography{polarwiretap}

\begin{thebibliography}{10}
\providecommand{\url}[1]{#1}
\csname url@samestyle\endcsname
\providecommand{\newblock}{\relax}
\providecommand{\bibinfo}[2]{#2}
\providecommand{\BIBentrySTDinterwordspacing}{\spaceskip=0pt\relax}
\providecommand{\BIBentryALTinterwordstretchfactor}{4}
\providecommand{\BIBentryALTinterwordspacing}{\spaceskip=\fontdimen2\font plus
\BIBentryALTinterwordstretchfactor\fontdimen3\font minus
  \fontdimen4\font\relax}
\providecommand{\BIBforeignlanguage}[2]{{%
\expandafter\ifx\csname l@#1\endcsname\relax
\typeout{** WARNING: IEEEtran.bst: No hyphenation pattern has been}%
\typeout{** loaded for the language `#1'. Using the pattern for}%
\typeout{** the default language instead.}%
\else
\language=\csname l@#1\endcsname
\fi
#2}}
\providecommand{\BIBdecl}{\relax}
\BIBdecl

\bibitem{chou2022isit}
R.~A. Chou, ``Quantifying the cost of privately storing data in distributed
  storage systems,'' in \emph{IEEE International Symposium on Information
  Theory (ISIT)}, 2022, pp. 3287--3292.

\bibitem{shamir1979share}
A.~Shamir, ``How to share a secret,'' \emph{Communications of the ACM},
  vol.~22, no.~11, pp. 612--613, 1979.

\bibitem{blakley}
G.~Blakley, ``Safeguarding cryptographic keys,'' \emph{Proceedings of the
  National Computer Conference}, pp. 313--317, 1979.

\bibitem{rawat2018centralized}
A.~S. Rawat, O.~O. Koyluoglu, and S.~Vishwanath, ``Centralized repair of
  multiple node failures with applications to communication efficient secret
  sharing,'' \emph{IEEE Transactions on Information Theory}, vol.~64, no.~12,
  pp. 7529--7550, 2018.

\bibitem{agarwal2016security}
A.~Agarwal and A.~Mazumdar, ``Security in locally repairable storage,''
  \emph{IEEE Transactions on Information Theory}, vol.~62, no.~11, pp.
  6204--6217, 2016.

\bibitem{soleymani2020distributed}
M.~Soleymani and H.~Mahdavifar, ``Distributed multi-user secret sharing,''
  \emph{IEEE Transactions on Information Theory}, vol.~67, no.~1, pp. 164--178,
  2020.

\bibitem{huang2016communication}
W.~Huang, M.~Langberg, J.~Kliewer, and J.~Bruck, ``Communication efficient
  secret sharing,'' \emph{IEEE Transactions on Information Theory}, vol.~62,
  no.~12, pp. 7195--7206, 2016.

\bibitem{bitar2017staircase}
R.~Bitar and S.~El~Rouayheb, ``Staircase codes for secret sharing with optimal
  communication and read overheads,'' \emph{IEEE Transactions on Information
  Theory}, vol.~64, no.~2, pp. 933--943, 2017.

\bibitem{shah2015distributed}
N.~B. Shah, K.~Rashmi, and K.~Ramchandran, ``Distributed secret dissemination
  across a network,'' \emph{IEEE Journal of Selected Topics in Signal
  Processing}, vol.~9, no.~7, pp. 1206--1216, 2015.

\bibitem{huang2016secure}
W.~Huang and J.~Bruck, ``Secure {RAID} schemes for distributed storage,'' in
  \emph{IEEE International Symposium on Information Theory (ISIT)}.\hskip 1em
  plus 0.5em minus 0.4em\relax IEEE, 2016, pp. 1401--1405.

\bibitem{huang2017secret}
------, ``Secret sharing with optimal decoding and repair bandwidth,'' in
  \emph{IEEE International Symposium on Information Theory (ISIT)}.\hskip 1em
  plus 0.5em minus 0.4em\relax IEEE, 2017, pp. 1813--1817.

\bibitem{chou2020secure}
R.~A. Chou and J.~Kliewer, ``Secure distributed storage: {R}ate-privacy
  trade-off and {XOR}-based coding scheme,'' in \emph{IEEE International
  Symposium on Information Theory (ISIT)}, 2020, pp. 605--610.

\bibitem{bessani2013depsky}
A.~Bessani, M.~Correia, B.~Quaresma, F.~Andr{\'e}, and P.~Sousa, ``Dep{S}ky:
  {D}ependable and secure storage in a cloud-of-clouds,'' \emph{ACM
  Transactions on Storage}, vol.~9, no.~4, pp. 1--33, 2013.

\bibitem{shor2018best}
R.~Shor, G.~Yadgar, W.~Huang, E.~Yaakobi, and J.~Bruck, ``How to best share a
  big secret,'' in \emph{Proceedings of the 11th ACM International Systems and
  Storage Conference}, 2018, pp. 76--88.

\bibitem{fabian2015collaborative}
B.~Fabian, T.~Ermakova, and P.~Junghanns, ``Collaborative and secure sharing of
  healthcare data in multi-clouds,'' \emph{Information Systems}, vol.~48, pp.
  132--150, 2015.

\bibitem{huang2017secure}
W.~Huang and J.~Bruck, ``Secure {RAID schemes from EVENODD and STAR} codes,''
  in \emph{IEEE International Symposium on Information Theory (ISIT)}.\hskip
  1em plus 0.5em minus 0.4em\relax IEEE, 2017, pp. 609--613.

\bibitem{karnin1983secret}
E.~Karnin, J.~Greene, and M.~Hellman, ``On secret sharing systems,'' \emph{IEEE
  Transactions on Information Theory}, vol.~29, no.~1, pp. 35--41, 1983.

\bibitem{mceliece1981sharing}
R.~J. McEliece and D.~V. Sarwate, ``On sharing secrets and {R}eed-{S}olomon
  codes,'' \emph{Communications of the ACM}, vol.~24, no.~9, pp. 583--584,
  1981.

\bibitem{benaloh1988generalized}
J.~Benaloh and J.~Leichter, ``Generalized secret sharing and monotone
  functions,'' in \emph{Conference on the Theory and Application of
  Cryptography}.\hskip 1em plus 0.5em minus 0.4em\relax Springer, 1988, pp.
  27--35.

\bibitem{ito1989secret}
M.~Ito, A.~Saito, and T.~Nishizeki, ``Secret sharing scheme realizing general
  access structure,'' \emph{Electronics and Communications in Japan (Part III:
  Fundamental Electronic Science)}, vol.~72, no.~9, pp. 56--64, 1989.

\bibitem{yamamoto1986secret}
H.~Yamamoto, ``Secret sharing system using (k, {L}, n) threshold scheme,''
  \emph{Electronics and Communications in Japan (Part I: Communications)},
  vol.~69, no.~9, pp. 46--54, 1986.

\bibitem{blakley1984security}
G.~R. Blakley and C.~Meadows, ``Security of ramp schemes,'' in \emph{Workshop
  on the Theory and Application of Cryptographic Techniques}.\hskip 1em plus
  0.5em minus 0.4em\relax Springer, 1984, pp. 242--268.

\bibitem{blundo1996randomness}
C.~Blundo, A.~De~Santis, and U.~Vaccaro, ``Randomness in distribution
  protocols,'' \emph{Information and Computation}, vol. 131, no.~2, pp.
  111--139, 1996.

\bibitem{blundo1993efficient}
C.~Blundo, A.~D. Santis, and U.~Vaccaro, ``Efficient sharing of many secrets,''
  in \emph{Annual Symposium on Theoretical Aspects of Computer Science}.\hskip
  1em plus 0.5em minus 0.4em\relax Springer, 1993, pp. 692--703.

\bibitem{yoshida2018optimal}
M.~Yoshida, T.~Fujiwara, and M.~P. Fossorier, ``Optimal uniform secret
  sharing,'' \emph{IEEE Transactions on Information Theory}, vol.~65, no.~1,
  pp. 436--443, 2018.

\bibitem{zou2015information}
S.~Zou, Y.~Liang, L.~Lai, and S.~Shamai, ``An information theoretic approach to
  secret sharing,'' \emph{IEEE Transactions on Information Theory}, vol.~61,
  no.~6, pp. 3121--3136, 2015.

\bibitem{chou2020unified}
R.~A. Chou, ``Unified framework for polynomial-time wiretap channel codes,''
  \emph{arXiv preprint arXiv:2002.01924}, 2020.

\bibitem{chou2018explicit}
R.~Chou, ``Explicit codes for the wiretap channel with uncertainty on the
  eavesdropper's channel,'' in \emph{IEEE International Symposium on
  Information Theory (ISIT)}, 2018, pp. 476--480.

\bibitem{rana2021information}
V.~Rana, R.~Chou, and H.~Kwon, ``Information-theoretic secret sharing from
  correlated {G}aussian random variables and public communication,'' \emph{IEEE
  Transactions on Information Theory}, vol.~68, no.~1, pp. 549--559, 2021.

\bibitem{chou2021distributed}
R.~Chou, ``Distributed secret sharing over a public channel from correlated
  random variables,'' \emph{arXiv preprint arXiv:2110.10307}, 2021.

\bibitem{chou2018secret}
R.~A. Chou, ``Secret sharing over a public channel from correlated random
  variables,'' in \emph{IEEE International Symposium on Information Theory
  (ISIT)}, 2018, pp. 991--995.

\bibitem{csiszar2010capacity}
I.~Csisz{\'a}r and P.~Narayan, ``Capacity of a shared secret key,'' in
  \emph{IEEE International Symposium on Information Theory (ISIT)}, 2010, pp.
  2593--2596.

\bibitem{sultana2021}
R.~Sultana and R.~Chou, ``Low-complexity secret sharing schemes using
  correlated random variables and rate-limited public communication,'' in
  \emph{IEEE International Symposium on Information Theory (ISIT)}, 2021, pp.
  970--975.

\bibitem{de1999multiple}
A.~De~Santis and B.~Masucci, ``Multiple ramp schemes,'' \emph{IEEE Transactions
  on Information Theory}, vol.~45, no.~5, pp. 1720--1728, 1999.

\end{thebibliography}

\end{document}